\documentclass{article}
\usepackage{cite}
\usepackage[english]{babel}														 			
\usepackage{amsmath,amsfonts,amsthm}									
\usepackage[pdftex]{graphicx}														
\usepackage{epstopdf}																
\usepackage{subfig}
\usepackage{booktabs}
\usepackage{footnote}
\usepackage{bbm}
\usepackage{dsfont}
\usepackage{amssymb}

\usepackage{cellspace, booktabs}
\usepackage{xcolor}
\usepackage{multirow, makecell}
\usepackage{rotating}
\usepackage[margin=1.19in]{geometry}
\usepackage{natbib}
\usepackage[utf8]{inputenc}
\usepackage[affil-it]{authblk}

\theoremstyle{plain}
\newtheorem{lem}{Lemma}[section]

\theoremstyle{plain}
\newtheorem{prop}{Proposition}[section]

\theoremstyle{remark}

\title{Prediction of spatial functional random processes:\\
 Comparing functional and spatio-temporal kriging approaches}
 
\author[1,*]{Johan Strandberg}
\author[1]{Sara Sj\"ostedt de Luna}
\author[2]{Jorge Mateu}

\affil[1]{Department of Mathematics and Mathematical Statistics, Ume\aa \ University, Ume\aa, Sweden}
\affil[2]{Department of Mathematics, Universitat Jaume I, Castell\'on, Spain}
\affil[*]{Corresponding author: Johan Strandberg\\
E-mail: johan.strandberg@math.umu.se\\
Tel.: +46-90-786 88 26}

\date{\today}
\begin{document}

\maketitle

\begin{abstract}
In this paper, we present and compare functional and spatio-temporal (Sp.T.) kriging approaches to predict spatial functional random processes (which can also be viewed as Sp.T. random processes). Comparisons with respect to computational time and prediction performance via functional cross-validation is evaluated, mainly through a simulation study but also on two real data sets. We restrict comparisons to Sp.T. kriging versus ordinary kriging for functional data (OKFD), since the more flexible functional kriging approaches, pointwise functional kriging (PWFK) and functional kriging total model, coincide with OKFD in several situations. We contribute with new knowledge by proving that OKFD and PWFK coincide under certain conditions.

From the simulation study, it is concluded that the prediction performance for the two kriging approaches in general is rather equal for stationary Sp.T. processes, with a 
tendency for functional kriging to work better for small sample sizes and Sp.T. kriging
to work better for large sample sizes. For non-stationary Sp.T. processes, with a common deterministic time trend and/or time varying variances and dependence structure, OKFD performs better than Sp.T. kriging irrespective of sample size.  For all simulated cases, the computational time for OKFD was considerably lower compared to those for the Sp.T. kriging methods. 
\end{abstract}
{\bf Keywords:} Prediction, Spatial functional random processes, Functional kriging, Spatio-temporal kriging.

\section*{Acknowledgements}
This work was supported by the Swedish Research Council (Project id D0520301) and J.Mateu has been partially funded by grants MTM2016-78917-R from the Spanish Ministery of Science, and P1-1B2015-40 from University Jaume I.

\newpage


%
%






 


\section{Introduction}

In many fields, such as environmental, forestry, climatology, meteorology and medical sciences, the spatial variation of objects in the form of curves are of interest to study. It could e.g. be ocean temperature, salinity or other variables measured over time (or at different depths) at a set of spatial locations. With today's modern technology and huge storage capability, it is in principle possible to observe entire curves by recording them over a dense raster of time (depth) points. In particular it may be of interest to predict a curve at a new spatial location given that such curves have been observed at $n$ other  locations by utilizing the information inherent in the spatial dependence between curves. 

Kriging predictors have a long history of being used to predict objects at new locations based on information observed at a set of other locations, especially for objects that are real- or vector-valued, see e.g. \cite{chiles2009geostatistics}, \cite{cressie2015statistics1}, \cite{cressie2015statistics2}, and references therein. Functional kriging predictors, used when the objects are functions with infinite dimension, was initially discussed by \cite{goulard1993geostatistical}, and further proposed by e.g. \cite{giraldo2010continuous, giraldo2011ordinary} and \cite{nerini2010cokriging}. In these papers, the expected value of the curves is assumed to be independent of the spatial location, so called ordinary functional kriging. More recently,  e.g. \cite{caballero2013universal}, \cite{menafoglio2013universal}, \cite{ignaccolo2014kriging}, and \cite{reyes2015residual}
 has investigated functional kriging methods where the expected value of the curves may also depend on location.
 

A kriging predictor is  a weighted sum of the objects observed at the $n$ spatial locations,  defined to be the best linear unbiased predictor (BLUP) minimizing the mean squared prediction error. 
The optimal kriging weights are functions of the (spatial) dependence structure of the objects,  which in practice needs to be estimated. Typically estimators of the dependence structure rely on stationarity assumptions, unless parametric and distributional assumptions are made. 


Here, two kriging approaches to predict spatial functional random processes are compared. A functional random process 
is a process with stochastic functional objects (curves) $\chi_s=\chi_s(t), t\in T$ over the "time" domain $T$ at each spatial location $s \in D$.  Given that the process has been observed at  $n$ different locations, a curve at a new location $s_0$,  can be predicted by a {\it functional kriging approach}, i.e. as a linear combination of the $n$  observed curves. A spatial functional process can also be viewed as a spatio-temporal (Sp.T.) random process $\{Z(s,t)=\chi_s(t), (s,t)\in D\times T\}$, and hence, a {\it Sp.T. kriging approach} could also be used. The curve $\chi_{s_0}(t), t \in T$ would then be predicted at a dense grid of values over $T$, based on linear combinations of a time-grid of values over the observed curves. The question of which approach, functional or Sp.T. kriging, should be used to analyze a particular data set is an  important one (with no closed answer), as pointed out by \cite{delicado2010statistics}. In this paper we compare the two approaches with respect to prediction performance and computational time, mainly by a simulation study and on two real data sets. Prediction performance is evaluated by functional cross-validation. 

In Section \ref{preliminaries} notation and definitions are given. Section \ref{kriging} presents the functional and Sp.T. kriging approaches, including how to estimate the dependence structure. We also discuss how the functional kriging methods relate to each other, and under which circumstances they may coincide. In particular we state conditions under which the two functional kriging methods ordinary kriging for functional data and pointwise functional kriging coincide, with proofs given in Appendix \ref{pwfkEokfd}. A simulation study, comparing the two kriging approaches, is presented in Section \ref{STT}, see also Appendix \ref{AppendixTables}. Both kriging approaches are applied to the Canadian temperature data (previously analyzed e.g. by \cite{giraldo2009geostatistical}, \cite{giraldo2010continuous} and \cite{menafoglio2013universal}) and to salinity in seawater data (previously analyzed e.g. by \cite{reyes2015residual} and \cite{romano2015performance}), see Section \ref{applications_main}. A discussion and concluding remarks are found in Section \ref{conclusions}.

\section{Preliminaries}
\label{preliminaries}
A {\em spatial functional random process} $\{\chi_s: s\in D \subset \mathbf{R}^d\}$ \citep{giraldo2010continuous,delicado2010statistics}, 
is a process where, for each given $s \in D$, the observed random element is a functional random variable, $\chi_s$,  taking values in an infinite dimensional space (or function space). We will consider the case where $\chi_s$ for every fixed $s$ is a real-valued function, $\chi_s(t), \; t \in T \subset \mathbf{R}$,  from the compact set $T$ to $\mathbf{R}$ and with $s\in D\subset \mathbf{R}^2$. It is usually assumed that the realizations of the curves (functions)  $\chi_s(t), t \in T, s\in D$ belong to a separable Hilbert space $\mathbf H$ of square integrable functions defined on $T$. 

A spatial functional random process is second-order stationary if for each $t\in T$ the corresponding spatial random process $\{\chi_s(t), s\in D\}$ is second-order stationary. 
For (second-order) stationary functional processes,  the covariance function (covariogram) satisfies $Cov[\chi_{s}(r),\chi_{v}(t)]=C(s-v, r, t)$, which can be described by the variogram 
\begin{equation*}
V[\chi_{s}(r)-\chi_{v}(t)]=2\gamma(s-v, r, t),
\end{equation*}
 via the relation
\begin{equation}
C(s-v, r, t)=\sigma^2(r)+\sigma^2(t)-\gamma(s-v, r, t),
\label{CovVar}
\end{equation}
where $\sigma^2(t)=V[\chi_{s}(t)]$. Our main focus is on {\em second-order isotropically stationary} spatial functional random processes, satisfying
\begin{itemize}
\item[(i)] $E[\chi_s(t)]=m(t)$ and $V[\chi_s(t)]=\sigma^2(t)$ $\forall s\in D$ and $\forall t\in T,$ \hfill\refstepcounter{equation}(\theequation) \label{def1}
\item[(ii)] $Cov[\chi_{s}(r),\chi_{v}(t)]=C(\|s-v\|, r, t)$ $\forall s, v\in D$ and $\forall r, t \in T$, 
\end{itemize}
where $\|\cdot\|$ denotes the (Euclidean) distance measure. For any given $t \in T$, 
$\gamma_t(h):=V[\chi_{s}(t)-\chi_{v}(t)]/2$, $h=\|s-v\|$, is the semivariogram of the spatial random process $\{\chi_s(t): s \in D\}$. In order to ensure that $V[\sum_{i=1}^nl_i\chi_{s_i}(t)]\ge0$ for any set of constants $l_1,\ldots, l_n \in R, n=1,2,... $, the variogram (as a function of $h$) needs to be a conditional negative definite function and the covariogram needs to be a positive definite function, see e.g. \cite{cressie2015statistics1}.

A spatial functional random process can also be viewed as a Sp.T. process $Z(s,t)=\chi_s(t)$, where $Z(s,t)$ takes values in $\mathbf{R}$, and is mapped from $(s,t) \in D \times T$, cf. \cite{cressie2015statistics2}. A Sp.T. process is 
said to be {\em second-order stationary and spatially isotropic} if 
\begin{itemize}
\item[(i)] $E[Z(s,t)]=m$ and $V[Z(s,t)]=\sigma_Z^2$ $\forall s\in D$ and $\forall t\in T,$\hfill\refstepcounter{equation}(\theequation) \label{def1SpT}

\item[(ii)] $Cov[Z(s,r),Z(v,t)]=C_Z(\|s-v\|, \mid r-t\mid)$ $\forall s, v\in D$ and $\forall r, t \in T$. \end{itemize}
Note that the class of stationary Sp.T processes is a subset of the class of stationary functional random processes, since
a  stationary Sp.T. process implies that the corresponding functional random process also is stationary, while the opposite may not be true.

\section{Kriging prediction}
\label{kriging}
In this section two kriging approaches to predict spatial functional random processes are described. Section \ref{functional_methods} presents different functional kriging methods, and under which circumstances they may coincide. Section \ref{sptk} describes the Sp.T. kriging approach. A way to evaluate prediction performance using functional cross-validation is given in Section \ref{evaluation}. 

\subsection{Functional kriging}
\label{functional_methods}


For the presentation below, unless otherwise stated, we will assume that the spatial functional random process is second-order stationary and isotropic.  Within the functional kriging framework, it is of interest to predict the complete random function $\chi_{s_0}(t), t \in  T$, at a new location $s_0$, given that a sample of random functions has been observed at $n$ different locations, $s_1, \ldots, s_n$. 
A functional kriging predictor, 
$\hat{\chi}_{s_0}(t), \; t \in  T$, is defined to be the best linear unbiased predictor (BLUP) minimizing the mean integrated squared error (MISE)
\begin{equation}
\text{MISE}(s_0)=E \Big\lbrack\int_{T} (\hat{\chi}_{s_0}(t)- \chi_{s_0}(t))^2dt\Big\rbrack.
\label{eqn:mise}
\end{equation}

\subsubsection{Ordinary kriging for functional data}
\label{okfd}
\cite{goulard1993geostatistical} proposed one of the first functional kriging predictors, the so called {\em curve kriging predictor} 
\begin{equation}
\hat{\chi}_{s_0}(t)=\sum_{i=1}^n\lambda_i \chi_{s_i}(t), \; \; t\in T.
\label{eqn:OKFD}
\end{equation}
This predictor was further discussed by \cite{giraldo2007geostatistics, giraldo2011ordinary} and there given the name of {\em ordinary kriging for functional data} (OKFD). 
The optimal kriging weights $\lambda_1,...,\lambda_n \in \mathbf{R}$  are chosen such that $\text{MISE}(s_0)$
is minimized given that the predictor is unbiased. It turns out that the optimal $\lambda_i$'s only depend on the {\em (isotropic) trace-semivariogram}, being defined as

\begin{equation}
\gamma(h_{ij})=\frac{1}{2}E\Big[ \int_{T} (\chi_{s_i}(t)-\chi_{s_j}(t))^2dt\Big]=\int_{T} \gamma_t(h_{ij})dt, \hspace{5pt} \forall s_i, s_j\in D,\label{ekv3}
\end{equation}
where $h_{ij}=\|s_i-s_j\|$. The second equality holds by Fubini's theorem under the assumption that the realizations of the random functions are square integrable. For a detailed derivation of the optimal weights, see \cite{giraldo2011ordinary}. The trace-semivariogram often satisfies the  properties of a classical semivariogram, being a conditional negative definite function  \citep{menafoglio2013universal}.



The trace-variogram is in practice unknown and therefore needs to be estimated from the data. Under assumption \eqref{def1},  a (consistent) method of moments estimator of the trace-semivariogram \eqref{ekv3} can be formed for a set of $h$-values as
\begin{equation}
\hat{\gamma}(h)=\frac{1}{2\vert N(h)\vert} \sum_{i,j\in N(h)}\int_{T}(\chi_{s_i}(t)-\chi_{s_j}(t))^2dt,
\label{ekv4}
\end{equation}
where $N(h)=\{ (s_i,s_j): \| s_i-s_j\|=h\},$ and $\vert N(h)\vert$ is the number of distinct elements in $N(h)$. For irregularly spaced observations, it is rare to have several pairs of observations separated at exactly distance $h$ and then  $N(h)$ is modified to $\{ (s_i,s_j): \| s_i-s_j\|\in(h-\epsilon,h+\epsilon)\},$ with $\epsilon>0$ being some small positive value, in order to obtain a more stable estimate. To obtain a valid (variogram) estimate for any $h$, a parametric variogram model $\gamma(h\mid\theta)$, e.g. the spherical, exponential or stable model, is fitted to a set of estimated values $\{\hat{\gamma}(h_l),h_l \}$, $l=1,...,L$, by a least squares method, cf. \cite{cressie2015statistics1}. Here, the ordinary least squares (OLS) method is used to estimate $\theta$. 

The random functions, $\chi_{s_i}(t)$, are typically observed only at a finite number of time points $t_{i1}, \ldots, t_{im_i},$ $i=1, \ldots, n$. \cite{goulard1993geostatistical} suggested to fit a parametric model $\chi_{s_i}(\cdot,\alpha_i)$ to the observed values and replace  $\chi_{s_i}(t)$ by $\chi_{s_i}(t,\hat{\alpha}_i)$ in \eqref{eqn:OKFD} and \eqref{ekv4}. A non-parametric approach was suggested by \cite{giraldo2011ordinary}, where the observed random functions  are represented (approximated) by linear combinations of $p$ known basis functions, ${\bf B}(t)=(B_1(t), \ldots, B_p(t))^{\intercal}$, as 
\begin{equation}
\tilde{\chi}_{s_i}(t)=\sum_{k=1}^p a_{ik}B_k(t)={\bf a}_i^{\intercal}{\bf B}(t).  
\label{eqn:Xbasisexp}
\end{equation}
The basis functions could e.g. be B-splines, Fourier basis, or wavelets. The coefficients ($\bf{a}_i$'s) can typically be determined by the least squares method. 
\cite{giraldo2011ordinary} suggested to chose the number of basis functions $p$ by  cross-validation. In the final  ordinary kriging predictor \eqref{eqn:OKFD}, the estimated trace-variogram values are plugged into the kriging weights ($\lambda_i$'s), and the $\tilde{\chi}_{s_i}(t)$'s replacing the 
$\chi_{s_i}(t)$'s.

\subsubsection{Pointwise functional kriging}
\label{secPWFK}

\cite{pwfk2008, giraldo2010continuous} suggested the {\it point-wise functional kriging predictor} (PWFK), to allow more flexibility than the OKFD predictor \eqref{eqn:OKFD}.  
It allows the $\lambda_i$'s to depend on $t$, and is defined as 
\begin{equation*}
\hat{\chi}_{s_0}(t)=\sum_{i=1}^n \lambda_i (t)\chi_{s_i}(t), \hspace{10pt} t \in T.  
\end{equation*}
The best linear unbiased predictor minimizing the mean squared integrated prediction error is found by choosing the $\lambda_i(t)$-functions such that (\ref{eqn:mise}) is minimized subject to the unbiasedness constraint of the predictor, $\sum_{i=1}^n \lambda_i(t)=1$, for all $t \in T.$ In order to solve the optimization problem, the $\lambda_i(t)$-functions are represented by a linear combination of $K$ known basis functions,
\begin{equation}
\lambda_{i}(t)=\sum_{k=1}^Kb_{ik}B_{\lambda k}(t)=\mathbf{b}_i^{\intercal}\mathbf{B}_{\lambda}(t), \hspace{10pt} i=1,...,n,
\label{eqn:Lbasisexp}
\end{equation}
where the $\mathbf{b}_i$'s are to be determined. Moreover, the $\chi_{s_i}(t)$'s are represented as in (\ref{eqn:Xbasisexp}), implying that $E[\chi_{s_i}(t)]=E[\mathbf{a}_i]^{\intercal}\mathbf{B}(t)$ and 
$Cov[\chi_{s_i}(t),\chi_{s_j}(u)]=\mathbf{B}(t)^{\intercal} Cov[ \mathbf{a}_i,\mathbf{a}_j]\mathbf{B}(u)$. The optimization problem then reduces the infinite dimensional problem to a multivariate geostatistics problem. Given that the weights satisfy \eqref {eqn:Lbasisexp}, the unbiasedness condition implies that  
\begin{equation}
\sum_{i=1}^n \lambda_i(t)=\sum_{i=1}^n\mathbf{b}_i^{\intercal}\mathbf{B}_{\lambda}(t)=\mathbf{c}^{\intercal}\mathbf{B}_{\lambda}(t)=1, \text{for all } t \in T,
\label{Blambdac}
\end{equation}
where $\mathbf{c}=\sum_{i=1}^n\mathbf{b}_i$. Hence, only basis functions $\mathbf{B}_{\lambda}(t)$ that satisfy \eqref{Blambdac} for some constant vector $\mathbf{c}$ give admissable solutions to the kriging optimization problem. When $\mathbf{B}_{\lambda}(t)$ are B-splines, \eqref{Blambdac} is fulfilled when $\mathbf{c}=\mathbf{1}$,  and for Fourier basis functions when $\mathbf{c}=(1, 0, \ldots, 0)^{\intercal}$. In fact any set of basis functions where one (the first say) basis function is a constant, $B_{\lambda 1}(t)=k$, satisfies \eqref{Blambdac} for $\mathbf{c}=(1/k, 0, \ldots, 0)^{\intercal}$. 
The full derivation of the equation system to be solved in order to find the $\mathbf{b}_i$'s, for admissable choices of $\mathbf{B}_{\lambda}(t)$ satisfying \eqref{Blambdac}, is given by \cite{giraldo2010continuous} when $\mathbf{B}_{\lambda}(t)=\mathbf{B}(t)$, and for general $\mathbf{B}_{\lambda}(t)$ in Appendix \ref{pwfkEokfd}.

The  $\mathbf{b}_i$'s turn out to be functions of the covariances between the various $\mathbf{a}_i$'s, which in practice are not known and thus need to be estimated. 
If $\mathbf{a}_i=\mathbf{a}(s_i)$, 
and $\mathbf{a}(s)=[a_1(s), \ldots, a_p(s)]^{\intercal}$
 is a $p$-variable second-order isotropically stationary spatial random field for all $s \in D$, with $E[\mathbf{a}(s)]=\mathbf{m}_a$ and $Cov[\mathbf{a}(s_i),\mathbf{a}(s_j)]=\Sigma(\|s_i-s_j\|) =\{c_{kl}(h_{ij})\}   \in \mathbf{R}^{p\times p}$, where $c_{kl}(h_{ij})=Cov[a_{k}(s_i),a_{l}(s_j)]$,  $h_{ij}=\|s_i-s_j\|$, it follows that
 $\{\chi_s(t) =\mathbf{a}(s)^{\intercal}\mathbf{B}(t), s \in D, t \in T\}$ satisfies \eqref{def1}.
 Under this assumption \cite{giraldo2010continuous} suggest estimating the covariograms and crosscovariograms (the $c_{kl}(\cdot)$'s) via a linear model of coregionalization \citep{goulard1992linear}. This means that $\mathbf{a}(s)$ can be expressed as $\mathbf{a}(s)=\mathbf{P}\mathbf{r}(s)$ where $\mathbf{P} \in \mathbf{R}^{p\times q} $ and $\mathbf{r}(s)=(r_1(s), \ldots, r_q(s))^{\intercal}$  are $q$ latent univariate (second-order isotropically stationary) random fields, typically assumed to be independent. 
 Given available data, $\mathbf{a}_i=\mathbf{a}(s_i), i=1,\ldots n$, the $c_{kl}(\cdot)$'s (and $\mathbf{P}$) can be estimated using the R-package gstat \citep{pebesma2004multivariable}. In order to perform the estimation, the value of $q$ and the variogram models of the ${r}_i(s)$'s need to be specified. 

The PWFK may have the potential to give better prediction performance than OKFD since it allows more flexible kriging weights. 
In which situations this could be true is still not completely known. In the following proposition (see Appendix \ref{pwfkEokfd} for the proof) we have confirmed situations in which PWFK and OKFD do coincide. 

\begin{prop}
\label{prop}
Suppose that the $\chi_{s_i}(t)$'s can be represented by (\ref{eqn:Xbasisexp}) and that $\mathbf{a}_i=\mathbf{a}(s_i)$ follows a linear model of coregionalization with $q$ independent second-order stationary latent univariate spatial random fields with common variogram. Further assume that the $\mathbf{B}_{\lambda}(t)$'s satisfy \eqref{Blambdac} for some constant vector $\mathbf{c}$, and that the inverse of the matrix $G$ exists, where
$$\mathbf{G}=\int_{T} \mathbf{B}^{\intercal}(t)\mathbf{P}\mathbf{P}^{\intercal}\mathbf{B}(t)\mathbf{B}_{\lambda}(t)\mathbf{B}_{\lambda}^{\intercal}(t)dt.$$
Then the optimal kriging weights \eqref{eqn:Lbasisexp} of PWFK that minimizes \eqref{eqn:mise} satisfy $\lambda_i(t)=\lambda_i$, with $\mathbf{b}_{i}=\lambda_i\mathbf{c}$, for all $i=1, \ldots, n$, and thus coincide with those of OKFD.

\end{prop}
\cite{giraldo2010continuous} kindly permitted us to use their R-code to estimate and predict PWFK models. On all simulated and real data sets we considered, the computational time for PWFK turned out to always be substantially larger than for OKFD, estimated by the R-package geofd \citep{Giraldo2012geofd}. Furthermore, we found a bug in their code, related to the numerical integration part. After correction of this bug, the estimated PWFK kriging weights always became constant (i.e. $\lambda_i(t)=\lambda_i$) when $\mathbf{B}_{\lambda}(t)=\mathbf{B}(t)$ were cubic B-splines or Fourier basis functions. 


\subsubsection{Functional kriging total model}
\label{secFKTM}
\cite{giraldo2009geostatistical, giraldo2014cokriging}, and independently \cite{nerini2010cokriging}, proposed a third functional kriging method, that allows to use all time points of the observed functions in the prediction of $\chi_{s_0}(t), \; t \in T$. The method is called the {\em functional kriging total model} (FKTM), and the predictor is defined as 
\begin{equation}
\hat{\chi}_{s_0}(t)=\sum_{i=1}^n \int_{T}\lambda_i (t, v)\chi_{s_i}(v)dv, \hspace{10pt} t \in T.  
\label{FKTM}
\end{equation}
This modeling approach is coherent with the functional linear model for functional responses (total model) introduced by \cite{ramsay2005functional}. Assuming that the random functions $\chi_{s_i}(t)$'s satisfy (\ref{eqn:Xbasisexp})  and that the kriging weights  satisfy
\begin{equation*}
\lambda_{i}(t,v)=\sum_{k=1}^p\sum_{l=1}^pc_{ik}^lB_{k}(t)B_{l}(v)=\mathbf{B}(t)^{\intercal}\mathbf{C}_i\mathbf{B}(v), \hspace{10pt} i=1,...,n,
\end{equation*}
\cite{giraldo2014cokriging} proposed a way to determine the $\lambda_i (t, v)$'s (i.e. the $\mathbf{C}_i$'s) such that the predictor (\ref{FKTM}) is unbiased and minimizes \eqref{eqn:mise}. Also here, the $\mathbf{C}_i$'s  turn out to be functions of the covariances between the various $\mathbf{a}_i$'s, which in practice are not known and can be estimated as proposed in Section \ref{secPWFK}. See \cite{giraldo2014cokriging} for more detailes.
The FKTM method is computationally heavy compared to OKFD, just like the PWFK method \citep{giraldo2009geostatistical}. Moreover, \cite{menafoglio2016kriging} showed that if the realizations of $\chi_s(t)$ belong to the Hilbert space of square integrable functions on $T$, and the functional second-order stationary random process is Gaussian, then the kriging weights of FKTM and OKFD agree a.s. for any orthonormal base $\mathbf{B}(t)$.


\subsection{Spatio-temporal kriging}
\label{sptk}
Since a spatial functional process also can be viewed as a Sp.T. process, $Z(s,t)=\chi_{s}(t)$, taking values in $(s,t) \in D \times T$, it could also be predicted by Sp.T. kriging methods. Given the observed values $Z(s_i,t_{ij}), j=1, \ldots, m_i, i=1, \ldots, n$, the Sp.T. kriging predictor at location $s_0$ and time point  $t \in T$, is defined to be the best linear unbiased predictor (BLUP), 
\begin{equation}
\hat{Z}(s_0,t)= \sum_{i=1}^n\sum_{j=1}^{m_i}\lambda_{ij} ^tZ(s_i,t_{ij}),
\label{SpTkrigingP}
\end{equation}
minimizing the mean squared prediction error (MSPE)
\begin{equation}
\text{MSPE}(s_0,t)=E[(\hat{Z}(s_0,t)- Z(s_0,t))^2].
\label{eqn:MSPE}
\end{equation}
Note that  for each $s_0$, the Sp.T. kriging weights ($\lambda_{ij}^t$'s) are allowed to change for each $t \in T$.  When the mean value of the process is constant, the unbiasedness condition implies that  $ \sum_{i=1}^n\sum_{j=1}^{m_i} \lambda_{ij}^t =1$. Moreover, if the constant mean value of the process is unknown, the kriging weights depend on the Sp.T. covariance structure solely, and \eqref{SpTkrigingP} is referred to as the so called  {\em Sp.T. ordinary kriging predictor}, see e.g. \cite{cressie2015statistics2}. The dependence structure in practice  needs to be estimated from the data and is then plugged into the kriging weights ($\lambda_{ij}^t$'s). For second-order stationary and spatially isotropic, Sp.T. processes satisfying \eqref{def1SpT}, the dependence structure, given by the (spatially isotropic) Sp.T. variogram,  
\begin{equation*}
E[(Z(s,r)-Z(v,t))^2]=2\gamma_Z (\|s-v\|,|r-t|), \hspace{5pt} s, v \in D \hspace{5pt} \text{and}\hspace{5pt} r, t \in T,
\end{equation*} 
is typically estimated via the following steps: First, an empirical (spatially isotropic) Sp.T. semivariogram is computed from lag classes as
\begin{equation*}
\hat{\gamma}_Z(h,u)=\frac{1}{2\vert N(h,u)\vert}\sum_{(i, j,k,l)\in N(h,u)}(Z(s_i,t_{ik})-Z(s_j,t_{jl}))^2, 
\end{equation*} 
where $N(h,u)=\{ (s_i,t_{ik}), (s_j,t_{jl}): \| s_i-s_j\|\in(h-\epsilon,h+\epsilon),$ and  $\vert t_{ik}-t_{jl}\vert\in(u-\delta,u+\delta)\},$ for some $\epsilon, \delta >0$, and $\vert N(h,u)\vert$ is the number of distinct elements in $N(h,u)$.  A parametric semivariogram model, ${\gamma}(h,u\vert \theta)$, is then fitted to a set of $\{\hat{\gamma}(h_l,u_l), (h_l,u_l)\},  l=1, \ldots, L$ by a least squares method. 
Three commonly used types of stationary Sp.T. semivariogram (covariogram) models to estimate the Sp.T.  dependence structure are the separable, product-sum and metric models.  \cite{graler2016spatio} show how Sp.T. ordinary kriging prediction can be performed with these three models using the R-package gstat.

 {\em The separable model} assumes that the Sp.T. covariance function can be modeled by the product of the spatial and the temporal covariance functions, 
\begin{equation}
C_Z(h,u)=C_s(h)C_t(u).
\label{Covseparable}
\end{equation}
This model has the computational advantage of being able to express the covariance matrix as the Kronecker product between two covariance matrices (space and time) which simplifies and speeds up the computation of its determinant and inverse. 
{\em The product-sum model} is an extension of the separable model, where the covariance function is of the form,
\begin{equation*}
C_Z(h,u)=kC_s(h)C_t(u)+C_s(h)+C_t(u),
\end{equation*}
with $k>0$. 
{\em The metric Sp.T. covariance model} is given by
\begin{equation*}
C_Z(h,u)=C_{joint}(\sqrt{h^2+(\kappa u)^2}).
\end{equation*}
To treat the spatial and temporal distances equally, the spatial and temporal dimensions are matched by an anisotropy parameter $\kappa$. 
Note that when $\kappa=1$,  this covariance model corresponds to an isotropic second-order stationary random process in $\mathbf{R}^3$.%

More generally, in Sp.T. kriging modeling, the process is often described as 
\[
Z(s,t)= \mu(s,t) + \epsilon(s,t),
\]
where $\mu(s,t)$ is a deterministic trend, and $\epsilon(s,t)$ is a mean zero Sp.T. random field, usually assumed stationary. 
The trend is typically modeled by
\begin{equation}
\mu(s,t)=\boldsymbol{\beta}^{\intercal}\textbf{x}(s,t),
\label{eqn:UKtrend}
\end{equation}
where $\textbf{x}(s,t) \in \mathbf{R}^M$ is a set of $M$ known covariates, often chosen to be polynomials of $s$ and $t$, and $\boldsymbol{\beta} \in \mathbf{R}^M$ is an unknown parameter to be determined.  When the Sp.T. process has a deterministic (unknown) non-constant trend of the form \eqref{eqn:UKtrend}, then the BLUP \eqref{SpTkrigingP} that minimizes \eqref{eqn:MSPE} is called the {\em Sp.T. universal kriging predictor}, and the kriging weights are functions of both the dependence structure and the covariates evaluated at the observed and predicted locations, see e.g.  \cite{cressie2015statistics2} Section 4.1.2, page 148. An iterative weighted least squares method may be used to estimate $\boldsymbol{\beta}$ and the Sp.T. variogram parameter $\theta$. Firstly, $\boldsymbol{\beta}$ would be estimated by the OLS method, minimizing 
\begin{equation*}
\sum_{i=1}^n \sum_{j=1}^{m_i} (Z(s_i,t_{ij})-\boldsymbol{\beta}^{\intercal}\mathbf{x}(s_i,t_{ij}))^2.
\end{equation*}
Based on the resulting regression residuals, the Sp.T. semivariogram is then estimated by fitting a parametric Sp.T. semivariogram model to the corresponding empirical Sp.T. semivariogram by a least squares method. The parameter $\boldsymbol{\beta}$ is then re-estimated using a weighted least squares method, taking into account the estimated dependence structure of the residuals \citep{cressie2015statistics1}. The dependence structure (variogram) is again estimated based on the updated residuals, and the whole procedure iterated until convergence. 
Note that if the deterministic trend only depends on time, such that $\mu(s,t)=m(t)$, the functional kriging methods do not need to specify and estimate the trend, whereas the Sp.T. kriging methods need to. 

\subsection{Evaluation of kriging methods}
\label{evaluation}
 {\em Functional cross-validation} (FCV) is a common way of evaluating the prediction performance of prediction methods for functional data, as suggested by \cite{giraldo2010continuous, giraldo2011ordinary}. In FCV,  the data from each observed spatial location is removed, one at a time, and then predicted at all observed time points by the prediction method using the observed functional data at the remaining locations. The mean squared prediction error (MSPE)  is computed as 
\begin{equation}
\text{MSPE}=\frac{1}{n}\sum_{i=1}^n \sum_{j=1}^{m_i}(Z(s_i,t_{ij})-\hat{Z}^{-i}(s_i,t_{ij}))^2/m_i,
\label{fcv}
\end{equation}
where  $\hat{Z}^{-i}(s_i,t_{ij})$ denotes the predicted value at location $(s_i,t_{ij})$ based on the functional data with the observations  $Z(s_i,t_{ij})$, 
$j=1, \ldots, m_i$ excluded.

\section{A simulation study}
\label{STT}
Here we present a simulation study that aims to shed light over the relative merits of Sp.T. and functional kriging, with particular focus on Gaussian second-order stationary functional processes in $\mathbf{R}^2$. Since the functional kriging methods OKFD, PWFK, and FKTM often coincide for such processes (see Sections \ref{secPWFK} and \ref{secFKTM}) we restrict our comparisons to Sp.T. kriging versus OKFD. We simulate data from Gaussian processes with three main types of covariance structures. The first two scenarios have stationary isotropic separable and non-separable covariance functions, respectively. The third scenario corresponds to second-order stationary functional (but non-stationary Sp.T.) processes with constant mean. For all three scenarios, several  different cases are simulated, with varying strengths of spatial and temporal dependence. All the (24) considered cases in the study are presented in Table \ref{tab_different_cases}, where the different parameters control the Sp.T. correlation structure in the three different main scenarios. 
\begin{table}[!htbp]
\small
\caption{The 24 different types (cases) of simulated Gaussian processes and their parameters: isotropic second-order stationary Sp.T. processes with separable (cases 1-9) and non-separable (cases 10-18) covariance functions, and second-order stationary functional (but non-stationary Sp.T.) processes (cases 19-24) with constant means. The larger the value of $\alpha$ and $\beta$ the weaker the spatial and temporal correlation, respectively.}

\centering
\begin{tabular}{cccc|cccc|cccc}

   \midrule
\multicolumn{4}{c|}{Generated data}& \multicolumn{4}{c|}{Generated data} & \multicolumn{4}{c}{Generated data} \\

  \multicolumn{1}{c}{Case} & \multicolumn{1}{c}{Type}  & $\alpha$ & $\beta$ &  \multicolumn{1}{c}{Case} & \multicolumn{1}{c}{Type}  & $\alpha$ & $\beta$ &  \multicolumn{1}{c}{Case}&  \multicolumn{1}{c}{Type} & \#bases (p) & $\alpha$  \\
  \midrule \addlinespace

 1& \multicolumn{1}{c}{\multirow{10}*{\rotatebox[origin=c]{90}{Separable}}} & & 0.1 & 10& \multicolumn{1}{c}{\multirow{10}*{\rotatebox[origin=c]{90}{Non-Separable}}} & & 0.1 & 19& \multicolumn{1}{c}{\multirow{10}*{\rotatebox[origin=c]{90}{Non-Stationary}}} & & 0.1  \\
                                                    2& & 0.1 & 1 &11& & 0.1 & 1 &20& & 7 & 0.5  \\
                                                    3& & & 10 &   12& & & 10&   21& & & 2  \\ \addlinespace
                                                    
                                                    4& & & 0.1 & 13& & & 0.1 & 22& & & 0.1   \\
                                                    5& & 0.5 & 1 &   14& & 0.5 & 1 &   23& & 15 & 0.5  \\
                                                    6& & & 10  &  15& & & 10  &  24& & & 2  \\ \addlinespace
                                                      
                                                    7& & & 0.1 &  16& & & 0.1 &  & & &  \\
                                                    8& & 2 & 1 & 17& & 2 & 1 & & &  &   \\
                                                    9& & & 10 & 18& & & 10 & & & &  \\ \addlinespace

\midrule
\end{tabular}
\label{tab_different_cases}
\normalsize
\end{table}
For each case in Table \ref{tab_different_cases}, three different sample sizes were considered; $small$ referring to $n=6\times 6$ spatial locations and $m=12$ time points, $medium$ referring to $n=6\times 6$ spatial locations and $m=50$ time points, and $large$ referring to $n=15\times 15$ spatial locations and $m=50$ time points. For each sample size, the number of time points were equally distributed on $[0,1]$ and the spatial locations were located on a regular grid in $[0,1]\times [0,1]$. Moreover, for cases 1-18, the presence of a deterministic time trend,  $m(t)=9+3\sin(2\pi t)$, was also investigated. For each case, sample size (and trend type for cases 1-18), 100 replicates were simulated. Figure \ref{sim_data} illustrates examples of simulated data for six of the cases, all with constant means. The three main scenarios  are now presented in more detail, together with the simulated results.
\begin{figure}[!htb]
\begin{center}
\centerline{}
\includegraphics[width=\textwidth]{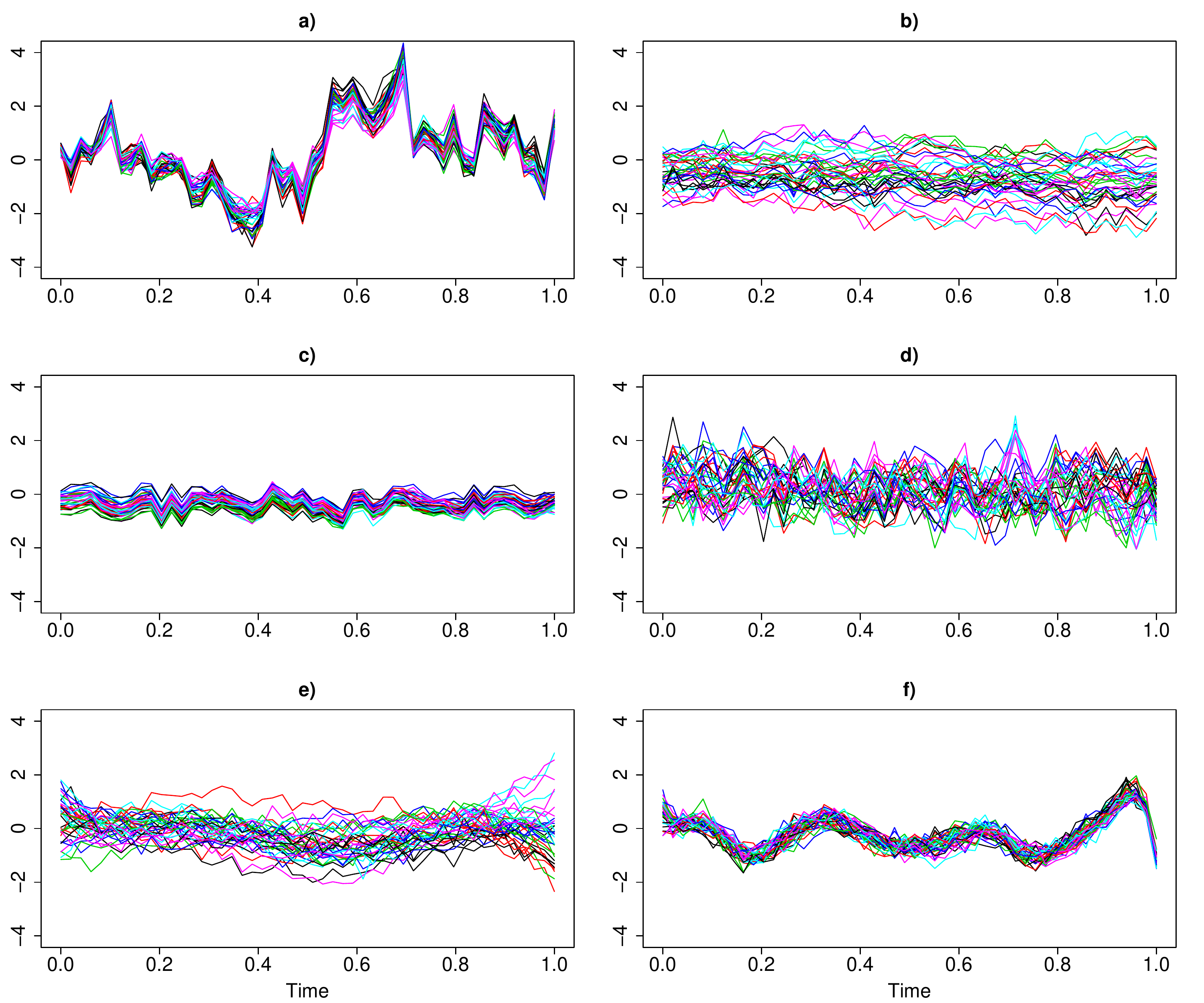} 
\caption{Examples of simulated data considering medium sample sizes without a deterministic time trend for: a) case 3 ($\alpha=0.1$, $\beta=10$), b) case 7 ($\alpha=2$, $\beta=0.1$), c) case 10 ($\alpha=0.1$, $\beta=0.1$), d) case 18 ($\alpha=2$, $\beta=10$), e) case 21 ($\alpha=2$, $p=7$) and f) case 22 ($\alpha=0.1$, $p=15$). The larger the value of $\alpha$ and $\beta$ the weaker the spatial and temporal correlation, respectively.}
\label{sim_data}
\end{center}
\end{figure} 

\subsection{Separable}
\label{sim_separable}

The first nine cases in Table \ref{tab_different_cases} were simulated (with and without the deterministic time trend) using the R-package RandomFields \citep{Schlather2015analysis}, and are Gaussian Sp.T processes with separable covariance functions \eqref{Covseparable}. The  spatial covariance function $C_s(h)$ in \eqref{Covseparable} was chosen to be the exponential covariance function with nugget effect, 
\begin{equation*}
C_s(h)= (1-\nu)\exp{(-\alpha h)}+\nu I\{h=0\}.
\end{equation*}
The nugget effect $\nu$ was set to $0.04$. For the parameter $\alpha$, we considered the values $0.1, 0.5$ and $2$, corresponding to the effective ranges $30$, $6$ and $1.5$ (very strong, medium and weak spatial correlation), respectively. The temporal covariance function $C_t(u)$ in \eqref{Covseparable} was given by the stable covariance function 
\begin{equation}
C_t(u)=\exp{(-(\beta u)^ \gamma)}.
\label{temp}
\end{equation}
Here, $\gamma$ was fixed to $0.5$, while the values for $\beta$ were $0.1, 1$, and $10$, corresponding to the effective ranges $90$, $9$ and $0.9$ (very strong, medium and weak temporal correlation), respectively.  

\subsubsection{Without a deterministic time trend}
\label{sim_separable_without_trend}

Estimation of the OKFD model was performed using the R-package geofd \citep{Giraldo2012geofd}. Given the generated data $Z(s_i,t_j), i=1, \ldots, n, j=1, \ldots, m$, the OKFD model was estimated using different types (Fourier and cubic B-splines) and number ($p$) of basis functions, see Table \ref{Nbasisstat} in Appendix \ref{AppendixTables} for a detailed specification. The ($p$) cubic B-splines were constructed based on $p-4$ equally distributed interior knots on the interval $[0,1]$. For each number and type of basis function three different semivariogram models (spherical, exponential and stable) were fitted to the empirical trace-semivariogram. For each case (1-9), a total of 36, 42 and 42 OKFD models ( 2 types of basis functions $\times$ \# different numbers of basis functions $\times$ \#trace-semivariograms) for small, medium and large sample sizes, respectively, were estimated and fitted to the data, evaluated by FCV (functional cross-validation) in terms of the MSPE \eqref{fcv}, and the minimum MSPE over the models registered. The {\em overall MSPE} for each case and sample size was computed as the average minimum MSPEs over the 100 replicates.

The Sp.T. kriging models were estimated using the R-packages gstat \citep{pebesma2004multivariable} and spacetime \citep{pebesma2012spacetime}. The Sp.T. semivariogram models (separable, product-sum and metric) described in Section \ref{sptk} were fitted to the empirical Sp.T. semivariogram, being all pairwise combinations of the exponential, spherical and stable variograms for the spatial (isotropic), temporal and joint variogram models. Hence, this resulted in 9 separable, 9 product-sum and 3 metric Sp.T. semivariogram models. All the Sp.T. kriging models were evaluated by FCV, minimum MSPE registered over the different models within each of the three groups of dependence structures, and the overall MSPE computed for each case (1-9) and sample size. The Sp.T. models with a product-sum and a metric covariance function were not evaluated for large size samples, due to the large computational time.


The overall MSPEs for the OKFD and the different Sp.T. kriging models for cases 1-9 considering medium sample sizes without a deterministic time trend are presented in Table \ref{tab2}. Corresponding results for small and large sample sizes are reported in Appendix \ref{AppendixTables}, Tables \ref{tab1} and \ref{tab3}. The numbers highlighted in red correspond to the smallest overall MSPE (per case) while the numbers in parentheses report the average computational time (for estimation and FCV) in seconds over all estimated models and replications (when run on a 3.5 GHz Intel Core i7 processor with 32 GB ram memory). The second last column presents the number of times, out of the 100 realisations, that OKFD had lower (minimum) MSPE than the Sp.T. separable model. The last column in Table \ref{tab2} reports p-values from paired two-sided t-tests comparing the overall MSPEs between the OKFD and the Sp.T. separable models, and thus reflects for which cases significant differences occur. 
\begin{table}[!htb]
\scriptsize
\caption{Prediction performance in terms of mean squared prediction errors (MSPEs) for the simulated cases 1-18 without a deterministic time trend, medium sample size. The smallest overall MSPE for each case is highlighted in red. The numbers in parentheses represent the average computational time in seconds over the corresponding estimated models and replications. The column \#Times represents the number of times, out of the 100 realisations, that OKFD had lower (minimum) MSPE than the Sp.T. separable model. The last column shows p-values from two-sided paired t-tests comparing the overall MSPEs between the OKFD and the Sp.T. separable models.}
\centering
\begin{tabular}{cccc|cccc|cc}
\midrule
\multicolumn{4}{c|}{Generated data} & \multicolumn{4}{c|}{overall MSPE} & \multicolumn{2}{c}{Comparison}\\
  \multicolumn{1}{c}{Case} & \multicolumn{1}{c}{Type}  & $\alpha$ & $\beta$ & OKFD & Sp.T. separable &Sp.T. product-sum&Sp.T. metric& \#Times & p-value \\
  \midrule \addlinespace
 1& \multicolumn{1}{c}{\multirow{10}*{\rotatebox[origin=c]{90}{Separable}}} & & 0.1 & {\bf\textcolor{red}{0.061}} (0.2)  & 0.062 (26.7) &0.064 (88.0)&0.083 (90.3)&27&0.552 \\  
                                                    2& & 0.1 & 1 & 0.068 (0.2) & {\bf\textcolor{red}{0.067}} (26.1)  &0.072 (89.0)&0.080 (89.9)&23&0.059 \\
                                                    3& & & 10 & 0.069 (0.2) & {\bf\textcolor{red}{0.066}} (24.7) &0.069 (92.2)&0.083 (90.3)&13&$<$0.001 \\ \addlinespace
                                                    
                                                    4& & & 0.1 & {\bf\textcolor{red}{0.134}} (0.2)  & 0.143 (23.6)  &0.149 (85.4)&0.204 (85.7)&56&$<$0.001 \\
                                                    5& & 0.5 & 1 & {\bf\textcolor{red}{0.131}} (0.2)  & 0.135 (29.8) &0.145 (102.9)&0.217 (98.5)&42&0.011 \\
                                                    6& & & 10 & 0.139 (0.2)& {\bf\textcolor{red}{0.137}} (27.6) &0.164 (103.3)&0.214 (95.7)&24&0.044 \\ \addlinespace
                                                      
                                                    7& & & 0.1 &{\bf\textcolor{red}{0.334}} (0.2) & 0.357 (29.0)  &0.353 (100.6)&0.416 (97.6)&64&$<$0.001 \\
                                                    8& & 2 & 1 & {\bf\textcolor{red}{0.368}} (0.2)  & 0.400 (29.3) &0.403 (106.1)&0.520 (99.2)&65&$<$0.001 \\
                                                    9& & & 10 & {\bf\textcolor{red}{0.372}} (0.2)   & 0.386 (28.7) &0.445 (104.3)&0.529 (97.9)&54&0.001 \\ \addlinespace

\hline \addlinespace
 10& \multicolumn{1}{c}{\multirow{10}*{\rotatebox[origin=c]{90}{Non-Separable}}} & & 0.1 & {\bf\textcolor{red}{0.066}} (0.2) & 0.067 (26.2) &0.070 (87.6)&0.101 (89.9)&38&0.050 \\
                                                      11& & 0.1 & 1 & 0.066 (0.2) & {\bf\textcolor{red}{0.065}} (25.7) &0.069 (87.5)&0.100 (89.6)&25&0.082  \\
                                                      12& & & 10 & 0.065 (0.2)  & {\bf\textcolor{red}{0.064}} (24.4) &0.067 (90.2)&0.082 (89.8)&27&0.075 \\ \addlinespace
                                                      
                                                      13& & & 0.1 & {\bf\textcolor{red}{0.128}} (0.2)  & 0.139 (25.8) &0.145 (92.8)&0.201 (92.9)&49&0.001 \\
                                                      14& & 0.5 & 1 & {\bf\textcolor{red}{0.134}} (0.2)  & 0.140 (24.2) &0.154 (90.2)&0.248 (87.0)&55&$<$0.001 \\
                                                      15& & & 10 &  {\bf\textcolor{red}{0.137}} (0.2)  & 0.140 (27.4) &0.155 (99.5)&0.257 (93.3)&41&0.006 \\ \addlinespace 
                                                      
                                                      16& & & 0.1 & {\bf\textcolor{red}{0.366}} (0.2)  & 0.398 (26.8)  &0.391 (95.9)&0.430 (92.4)&67&$<$0.001 \\
                                                      17& & 2 & 1 & {\bf\textcolor{red}{0.354}} (0.2)  & 0.390 (24.8) &0.386 (91.8)&0.476 (86.5)&63&$<$0.001 \\
                                                      18& & & 10 & {\bf\textcolor{red}{0.373}} (0.2)  & 0.391 (27.4) &0.402 (95.2)&0.579 (90.4)&52&0.003 \\ \addlinespace                                                 
\midrule
\end{tabular}
\label{tab2}
\normalsize
\end{table}

For cases 1-9, the Sp.T. separable kriging models in general performed better (lower overall MSPEs) than the Sp.T. product-sum and metric models, which is natural since the simulated data were generated from Sp.T. models with separable covariance functions. Still, the overall MSPE was often (significantly) lower for OKFD compared to the Sp.T. separable models, for small and medium sample sizes (Tables \ref{tab2} and \ref{tab1}). For large sample sizes, the estimated Sp.T. (separable) models often performed better than OKFD (Table \ref{tab3}). In general, the larger the sample size, the more likely it is that the estimated Sp.T. (separable) models perform better than OKFD. Studying the overall MSPEs in more detail reveals that the weaker the spatial correlation and the stronger the temporal correlation, the better the OKFD performs and the worse the Sp.T. separable model performs, regardless of the sample size. Case 3 for example, with strong spatial and weak temporal correlation, has significantly lower overall MSPE for the Sp.T. separable model compared to the OKFD model for medium and large sample sizes (Tables \ref{tab2} and \ref{tab3}). On the other hand, for case 7, with weak spatial and strong temporal correlation, the result is reversed. 
Moreover, from the registered computational times in Tables \ref{tab2}, \ref{tab1}, and \ref{tab3}, it can be concluded that prediction by and estimation of an OKFD model is substantially faster than the Sp.T. kriging models for cases 1-9, regardless of the sample size. The Sp.T. separable models had lower computational time compared to the Sp.T. product-sum and metric models, as expected (see, Section \ref{sptk}).

Figure \ref{res2} presents how the type and number of basis functions used in the OKFD model affects the prediction performance (minimum MSPE over the three trace-semivariogram models, averaged over the 100 realisations) for cases 3 and 7 considering medium sample sizes. The number of basis functions turns out to be an important factor for prediction performance, in general with smaller prediction error the more basis functions are used. On the other hand, the type of basis functions, Fourier or cubic B-splines, is of less importance. These findings are consistent with all cases (1-9) and for all considered sample sizes. 

\begin{figure}[!htb]
\begin{center}
\includegraphics[width=\textwidth]{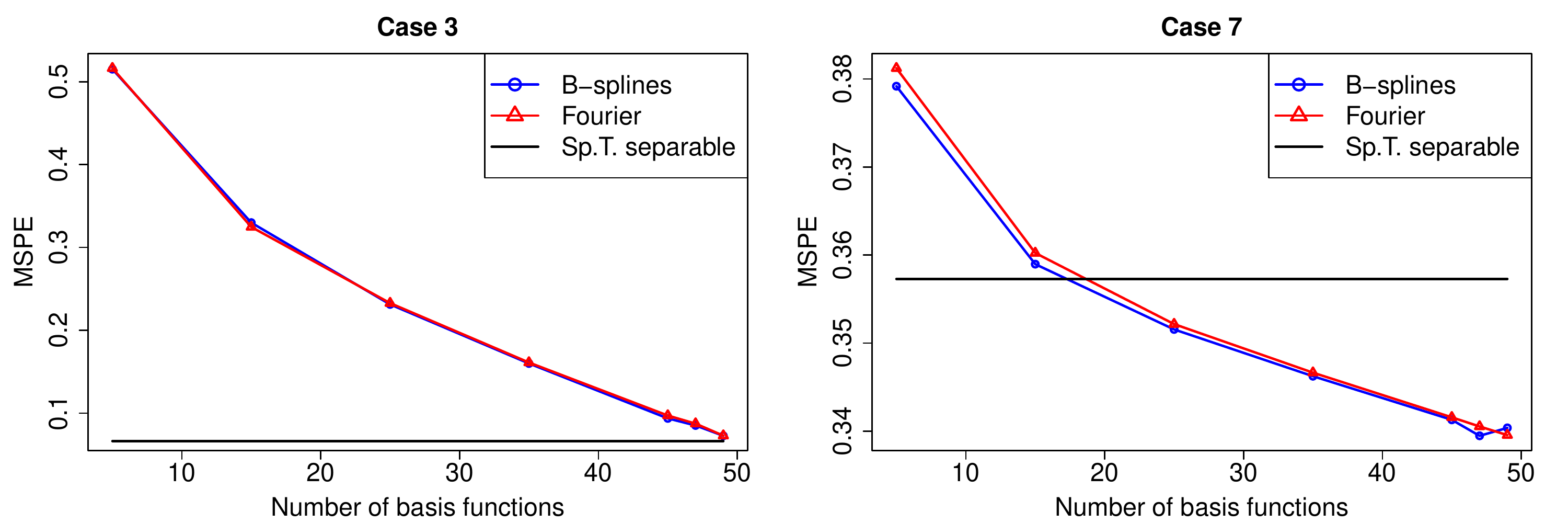}
\caption{Prediction performance (minimum MSPE over the three trace-semivariogram models, averaged over the 100 realisations) for cases 3 and 7 considering medium sample sizes without a deterministic time trend when the estimated OKFD model is based on different numbers ($p$) of basis functions, being both Fourier and cubic B-spline bases. The solid black lines represent the corresponding overall MSPE of the Sp.T. separable model.}
\label{res2}
\end{center}
\end{figure} 

To see how prediction performance may vary between replicates, Figure \ref{res1} presents box-plots of the differences in (minimum) MSPE between the two kriging approaches (MSPE(Sp.T)-MSPE(OKFD)) over the 100 replicates for cases 1-9 considering medium sample sizes. Here it becomes clear that OKFD produces more robust predictions. The Sp.T. kriging method (with estimated separable covariance function) produced much higher MSPEs than OKFD (casewise) for many realisations, especially for the cases with small and medium sample sizes.  

\begin{figure}[!htb]
\begin{center}
\includegraphics[width=\textwidth]{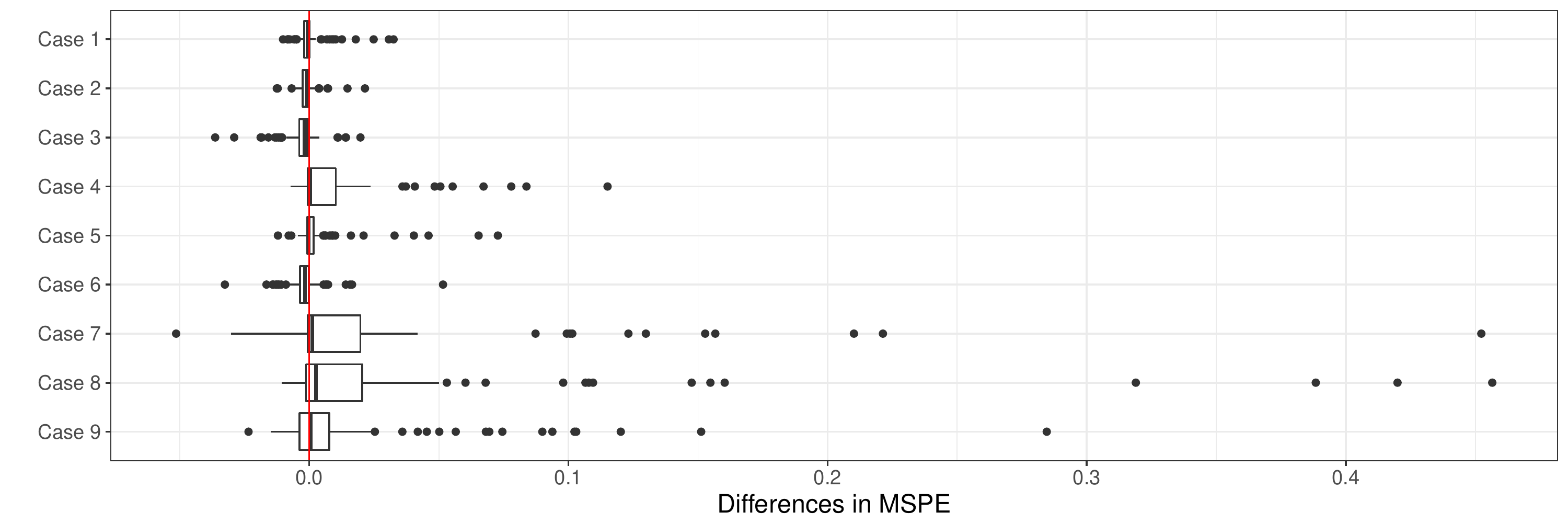} 
\caption{Box plots for cases 1-9 (considering medium sample sizes without a deterministic time trend) of the differences in (minimum) MSPE between the two kriging approaches (MSPE(Sp.T)-MSPE(OKFD)) for the 100 replicates.}
\label{res1}
\end{center}
\end{figure} 


\subsubsection{With a deterministic time trend}
\label{sim_separable_with_trend}
Simulated data sets for cases 1-9 with a common deterministic (sinusoidal) time trend were predicted by the same OKFD models as in Section \ref{sim_separable_without_trend}, since the OKFD models are designed to handle situations where a common deterministic time trend is present. Predictions were also made by universal Sp.T. kriging, using the same Sp.T. semivariogram models as in Section \ref{sim_separable_without_trend}. The deterministic time trend in the universal Sp.T. kriging model was specified to be the same as the one simulated from. 

Table \ref{tab5} summarizes the prediction performance of the two kriging approaches for cases 1-9  with deterministic time trend, for medium sample size. Corresponding results for small and large sample sizes are reported in Appendix \ref{AppendixTables}, Tables \ref{tab4} and \ref{tab6}. Comparing these tables with the corresponding tables in Section \ref{sim_separable_without_trend}, we see that the presence and estimation of a deterministic time trend did not have a large effect on the prediction performance, and more or less gave the same conclusions with respect to the relative performance of the two kriging approaches, regardless of the sample size. 

However, for small sample sizes, a difference showed up indicating that not only the number of basis functions, but also the type of basis functions used in the OKFD models matter. Here we noticed that, for any given number of basis functions used in the OKFD model, B-splines had better prediction performance than if Fourier basis functions was used. This is illustrated in Figure \ref{mspe_small_case3_7}, which presents how the type and number of basis functions used in the OKFD model affects the prediction performance (minimum MSPE over the three trace-semivariogram models, averaged over the 100 realisations) for cases 3 and 7 considering small sample sizes. We believe that a reason  for the observed effect is connected to that the functional representations based on B-splines better fits the data. The effect is not observed for larger sample sizes.

\begin{figure}[!htb]
\begin{center}
\includegraphics[width=\textwidth]{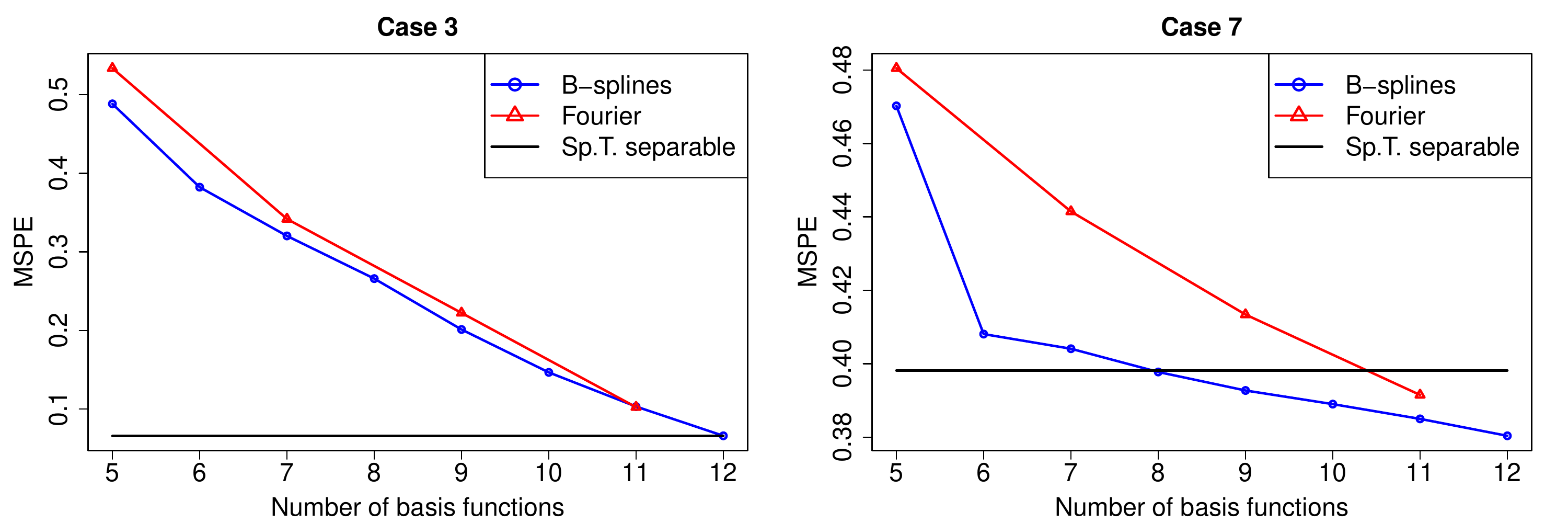}
\caption{Prediction performance (minimum MSPE over the three trace-semivariogram models, averaged over the 100 realisations) for cases 3 and 7 considering small sample sizes with a deterministic time trend when the estimated OKFD model is based on different numbers ($p$) of basis functions, being both Fourier and cubic B-spline bases. The solid black lines represent the corresponding overall MSPE of the Sp.T. separable universal kriging model.}
\label{mspe_small_case3_7}
\end{center}
\end{figure} 

\begin{table}[!htb]
\scriptsize
\caption{Prediction performance in terms of mean squared prediction errors (MSPEs) for the simulated cases 1-18 with a deterministic time trend, medium sample size. The smallest overall MSPE for each case is highlighted in red. The numbers in parentheses represent the average computational time in seconds over the corresponding estimated models and replications. The column \#Times represents the number of times, out of the 100 realisations, that OKFD had lower (minimum) MSPE than the Sp.T. separable model. The last column shows P-values from two-sided paired t-tests comparing the overall MSPEs between the OKFD and the Sp.T. separable models.}
\centering
\begin{tabular}{cccc|cccc|cc}
\midrule
\multicolumn{4}{c|}{Generated data} & \multicolumn{4}{c|}{overall MSPE} & \multicolumn{2}{c}{Comparison}\\
  \multicolumn{1}{c}{Case} & \multicolumn{1}{c}{Type}  & $\alpha$ & $\beta$ & OKFD & Sp.T. separable &Sp.T. product-sum&Sp.T. metric& \#Times & P-value \\
  \midrule \addlinespace
 1& \multicolumn{1}{c}{\multirow{10}*{\rotatebox[origin=c]{90}{Separable}}} & & 0.1 & {\bf\textcolor{red}{0.066}} (0.2)  & 0.066 (27.5)& 0.068 (89.4) & 0.094 (91.5)&29&0.706 \\
                                                    2& & 0.1 & 1 & 0.063 (0.2) & {\bf\textcolor{red}{0.063}} (26.9) & 0.066 (89.6) & 0.084 (91.2)&34&0.882 \\
                                                    3& & & 10 & 0.068 (0.2) & {\bf\textcolor{red}{0.066}} (25.8)  & 0.070 (93.8) & 0.090 (92.0)&19&$<$0.001  \\ \addlinespace
                                                                                                          
                                                    4& & & 0.1 & {\bf\textcolor{red}{0.135}} (0.2) & 0.141 (29.2) & 0.143 (95.7) & 0.220 (96.2)&49&0.002 \\
                                                    5& & 0.5 & 1 & {\bf\textcolor{red}{0.130}} (0.2) & 0.136 (30.4) & 0.153 (100.3) & 0.224 (96.0)&51&0.003 \\
                                                    6& & & 10 & 0.135 (0.2) & {\bf\textcolor{red}{0.133}} (28.5) & 0.157 (103.5) & 0.182 (96.0)&28&0.007 \\ \addlinespace
                                                      
                                                    7& & & 0.1 & {\bf\textcolor{red}{0.376}} (0.2) & 0.419 (29.1) & 0.411 (100.8) & 0.479 (96.2)&63&$<$0.001 \\
                                                    8& & 2 & 1 & {\bf\textcolor{red}{0.377}} (0.2)  & 0.394 (27.7) & 0.408 (95.8) & 0.524 (90.8)&65&0.001 \\
                                                    9& & & 10 & {\bf\textcolor{red}{0.372}} (0.2) & 0.385 (29.7) & 0.435 (103.4) & 0.531 (96.5)&44&0.002 \\ \addlinespace

\hline \addlinespace
 10& \multicolumn{1}{c}{\multirow{10}*{\rotatebox[origin=c]{90}{Non-Separable}}} & & 0.1 &  {\bf\textcolor{red}{0.063}} (0.2) & 0.064 (26.4) & 0.066 (88.5) & 0.099 (90.5)&27&0.463 \\
                                                      11& & 0.1 & 1 &  {\bf\textcolor{red}{0.065}} (0.2) & 0.065 (26.2) & 0.070 (88.1)  & 0.100 (90.8)&32&0.976 \\
                                                      12& & & 10 & 0.069 (0.2) & {\bf\textcolor{red}{0.067} }(25.8) & 0.072 (86.1) & 0.094 (86.6)&26&0.001 \\ \addlinespace
                                                      
                                                      13& & & 0.1 & {\bf\textcolor{red}{0.137}} (0.2) & 0.147 (27.8) & 0.152 (92.9) & 0.212 (93.5)&57&0.001 \\
                                                      14& & 0.5 & 1 & {\bf\textcolor{red}{0.138}} (0.2) & 0.144 (30.3) & 0.153 (97.4) & 0.270 (95.3)&41&0.034 \\
                                                      15& & & 10 & {\bf\textcolor{red}{0.136}} (0.2) & 0.140 (28.1) & 0.154 (98.1) & 0.230 (91.5)&41&0.002  \\ \addlinespace
                                                      
                                                      16& & & 0.1 & {\bf\textcolor{red}{0.359}} (0.2) & 0.406 (26.1) & 0.384 (96.0) & 0.450 (91.7)&84&$<$0.001 \\
                                                      17& & 2 & 1 & {\bf\textcolor{red}{0.367} }(0.2) & 0.414 (27.0) & 0.406 (96.2) & 0.516 (91.5)&74&$<$0.001 \\
                                                      18& & & 10 & {\bf\textcolor{red}{0.374}} (0.2) & 0.396 (27.5) & 0.398 (96.3) & 0.591 (91.9)&57&0.003  \\ \addlinespace                                             
\midrule
\end{tabular}
\label{tab5}
\normalsize
\end{table}

\subsection{Non-separable}
\label{sim_nonseparable}
Cases 10-18 in Table \ref{tab_different_cases} with and without the common deterministic time trend were simulated using the R-package RandomFields, and correspond to Gaussian Sp.T. processes with non-separable covariance functions of the form
\begin{equation*}
Cov_{\text {NSEP}}(h,u)= (1-\nu)(2- C_t(u) )^{-\delta /2} \exp{\left(-\frac{\alpha h}{\sqrt{2-C_t(u)}}\right)} + \nu I\{h=0\}.
\end{equation*}
with parameters set to $\delta=2$,  $\nu=0.04$, and $\alpha=0.1, 0.5$, and $2$. The covariance function $C_t(u)$ was chosen to be the stable covariance function  \eqref{temp} with $\gamma=0.5$ and $\beta=0.1, 1$, and $10$. The OKFD and the Sp.T. kriging models  estimated in Section \ref{sim_separable} were also fitted to the simulated data sets of cases 10-18 using the R-packages geofd, gstat, and spacetime, each with 100 realisations. 

Prediction performance of the two kriging approaches was evaluated in the same way as described in Section \ref{sim_separable_without_trend} and is summarized in Tables \ref{tab2} and \ref{tab5} for the non-separable cases 10-18 for medium sample sizes without and with a deterministic time trend, respectively. Corresponding results for small and large sample sizes are reported in Appendix \ref{AppendixTables}, Tables \ref{tab1}-\ref{tab6}. In general, we draw similar conclusions as in Section \ref{sim_separable} for the separable cases 1-9;  the Sp.T. separable kriging models perform better than the Sp.T. product-sum and metric models; the weaker the spatial correlation and the stronger the temporal correlation, the better the OKFD performs and the worse the Sp.T. (separable) model performs; OKFD works better for smaller sample sizes whereas fitted Sp.T. separable kriging models perform better for large sample sizes; more basis functions in OKFD generally improve prediction performance; computational times are much lower for OKFD; the presence of a deterministic time trend does not change the conclusions. 

A more detailed comparison of the overall MSPEs (and p-values) in Tables \ref{tab2}-\ref{tab5}, and Tables \ref{tab1}-\ref{tab6} reveals that prediction performance of OKFD in general improves in comparison to the Sp.T. separable kriging models for the simulated data sets with non-separable covariance functions (cases 10-18) compared to those simulated from separable covariance functions (cases 1-9). This result was to be expected, since none of the fitted (Sp.T.) kriging models coincide with the models that generated the data for cases 10-18.  

\subsection{Non-stationary}
\label{sim_nonstationary}
Generation of simulated data sets of second-order isotropic stationary functional, but non-stationary Sp.T. Gaussian processes with constant mean (cases 19-24 in Table \ref{tab_different_cases}) were based on the model
\begin{equation}
\label{nonstatsim}
\chi_{s_i}(t)={\bold{a}_{i}}^{\intercal} \bold{B}(t) +\epsilon_{s_i}(t),  \hspace{5pt} i=1,...,n.
\end{equation}
The basis functions $\bold{B}(t) \in \bold{R}^p$ with $p=7$, and $15$ cubic B-splines, are defined on equally space knots on the interval $[0,1]$. Moreover, $\bold{a}_{i}=(a_1(s_i), \ldots, a_p(s_i))^{\intercal}$, where $a_k(s), k=1, \ldots, p$, were chosen to be $p$ independent identically distributed second-order stationary isotropic zero mean Gaussian processes in $\mathbf{R}^2$ with exponential covariance function $C(h)=\exp(-\alpha h)$ with $\alpha=0.1, 0.5$ and $2$. Hence, the vectors  $(a_k(s_1), \ldots, a_k(s_n))^{\intercal}$, $k=1, \ldots, p$, are $p$ independent realisations of a multivariate Gaussian random variable $N_{n}(\bold{0},\bold{\Sigma})$, where the $n\times n$ covariance matrix equals $\bold{\Sigma}= \{\exp(-\alpha \|s_i-s_j\|)\}$. The $\epsilon_{s_i}(t)$'s are white noise measurement errors, independent and identically normally distributed  with mean 0.04 and variance 1, i.e. $\epsilon_{s_i}(t)\sim N(0.04,1)$. For each of  the $ 2 \times 3=6$ cases (19-24), 100 independent realisations were generated using the R-package fda \citep{ramsay2009functional}. 

To each generated data set we fitted the same OKFD models as those fitted in Section \ref{sim_separable_without_trend} using the R-package geofd. However, for medium and large sample sizes we extended the choices of number of basis functions (see Appendix \ref{AppendixTables} Table \ref{Nbasisnonstat} for specification), yielding a total of 36, 90 and 90 different estimated OKFD models for small, medium and large sample sizes, respectively. For each case (19-24), sample size and realisation, predictions were made and evaluated by FCV for all models, and the minimum MSPE over the models registered. The {overall MSPE} for each case and sample size was finally computed as the average minimum MSPE over the 100 replicates. Furthermore, the same Sp.T. ordinary kriging models fitted to the data in Section \ref{sim_separable_without_trend}, were also estimated for these data sets. Additionally, Sp.T. universal kriging models were fitted, with a deterministic time trend specified by a linear combination of the same basis functions that were used to generate the data set. Hence, a total of 18 separable, 18 product-sum and 6 metric Sp.T. kriging models were fitted to the data; predictions evaluated by FCV, the minimum MSPE registered over the models within the three groups of dependence structures, and the overall MSPE computed for each case (19-24) and sample size. As in Section \ref{sim_separable} and  \ref{sim_nonseparable}, the Sp.T. models with a product-sum and a metric covariance function were not evaluated for large sample sizes due to the large computational times.

Table \ref{table1} summarizes the prediction performance of the two kriging approaches for cases 19-24 and all three sample sizes. Note that these simulated data sets have time varying variances and covariances, which the Sp.T. kriging approach is not designed to capture, whereas the OKFD model can handle such situations. We would therefore expect OKFD to perform better than the Sp.T. kriging approach, which is indeed the case. In fact, OKFD has significantly lower overall MSPE for all cases and sample sizes in Table  \ref{table1} except for cases 22-23 considering small sample sizes. For these two cases the Sp.T. separable kriging model works better. This is probably coupled to the low number of observations (12) per location for small sample sizes. When the functional representations of the data at each location is formed for the OKFD models, we can thus at most fit a linear combination of 12 basis functions, whereas, the data are generated by 15 B-splines. The functional representations may thus fail to capture the full temporal time dynamics. The Sp.T. universal kriging models on the other hand, can fit a common deterministic time trend using all 15 B-splines. 
From Table \ref{table1}, it is also noted that Sp.T. kriging models with fitted metric variograms sometimes had better prediction performance than the Sp.T. separable kriging models, but still worse than the best OKFD models. Moreover, we again note that the computational time for OKFD is much lower than for the Sp.T. models.

\begin{table}[!htbp]
\tiny

\caption{Prediction performance in terms of mean squared prediction errors (MSPEs) for the simulated cases 19-24 over the different sample sizes. The smallest overall MSPE for each case is highlighted in red. The numbers in parentheses represent the average computational time in seconds over the corresponding estimated models and replications. The column \#Times represents the number of times, out of the 100 realisations, that OKFD had lower (minimum) MSPE than the best Sp.T. model. The last column shows p-values from two-sided paired t-tests comparing the overall MSPEs between the OKFD and the best Sp.T. model.}

\centering
\begin{tabular}{>{\centering\arraybackslash}p{0.8cm} >{\centering\arraybackslash}p{0.4cm}  c >{\centering\arraybackslash}p{1.2cm} >{\centering\arraybackslash}p{0.3cm}|cccc|cc}
   \hline\hline
\multicolumn{5}{c|}{Generated data} & \multicolumn{4}{c|}{overall MSPE} & \multicolumn{2}{c}{Comparison}\\
  \multicolumn{1}{c}{Scenario} & \multicolumn{1}{c}{Type} & \multicolumn{1}{c}{Data size} & \#bases (p) & $\alpha$ & OKFD & Sp.T. Separable & Sp.T. Product-sum & Sp.T. Metric & \#Times & p-value \\
  \hline \addlinespace
 19& \multicolumn{1}{c}{\multirowcell{23}{\rotatebox[origin=c]{90}{Non-stationary}}} & \multirow{8}{*}{Small} & & 0.1 &  {\bf\textcolor{red}{0.054}} (0.2) & 0.056 (8.0) & 0.057 (11.5) & 0.056 (6.3) & 87 & $<$0.001  \\  
                                                    20& & & 7 & 0.5 & {\bf\textcolor{red}{0.096}} (0.2) & 0.099 (8.4)  & 0.100 (14.1) & 0.100 (6.1) &100& $<$0.001\\
                                                    21& & & & 2 &  {\bf\textcolor{red}{0.226}} (0.2)  & 0.232 (9.0) & 0.236 (15.7) & 0.235 (6.4) & 92 & $<$0.001 \\ \addlinespace
                                                    
                                                    22& & & & 0.1 &  0.058 (0.2) & {\bf\textcolor{red}{0.057}} (7.4) & 0.058 (13.2) & 0.061 (6.1) & 11 & $<$0.001  \\
                                                    23& & & 15 & 0.5 & 0.099 (0.2) &  {\bf\textcolor{red}{0.099}} (7.8) & 0.101 (15.2) & 0.102 (6.0) & 41 &0.342\\
                                                    24& & & & 2 & {\bf\textcolor{red}{0.239}} (0.2) & 0.243 (8.7) & 0.260 (16.4) & 0.263 (6.5) & 45 &0.049\\ \addlinespace
                                                    
						     19& & \multirow{8}{*}{Medium} & & 0.1 &  {\bf\textcolor{red}{0.050}} (0.2)   & 0.055 (24.8) & 0.056 (83.8) & 0.052 (84.1) & 97 & $<$0.001 \\  
                                                    20& & & 7 & 0.5 & {\bf\textcolor{red}{0.083}} (0.2)  & 0.092 (24.4) & 0.093 (79.5) & 0.088 (79.6) &95& $<$0.001 \\
                                                    21& & & & 2 & {\bf\textcolor{red}{0.202}} (0.2)  & 0.212 (28.7) & 0.220 (91.2) & 0.210 (88.0) & 85&$<$0.001 \\ \addlinespace
                                                    
                                                    22& & & & 0.1 &  {\bf\textcolor{red}{0.052}} (0.2)  & 0.056 (26.1) & 0.056 (81.9) & 0.056 (84.2) &100& $<$0.001\\
                                                    23& & & 15 & 0.5 &  {\bf\textcolor{red}{0.087}} (0.2)  & 0.094 (28.0) & 0.093 (90.2) & 0.093 (87.3) & 100&$<$0.001\\
                                                    24& & & & 2 & {\bf\textcolor{red}{0.209}} (0.2) & 0.218 (28.2) & 0.229 (92.6) & 0.223 (87.7) & 100 &  $<$0.001\\ \addlinespace
                                                      
                                                    19& &  \multirow{8}{*}{Large} & & 0.1 & {\bf\textcolor{red}{0.044}} (9.0) & 0.047 (150.7) & & & 100 &$<$0.001  \\
                                                    20& & & 7 & 0.5 &  {\bf\textcolor{red}{0.055}} (9.1)  & 0.061 (151.6) &  &  &100 &$<$0.001\\
                                                    21& & & & 2 &   {\bf\textcolor{red}{0.097}} (9.2)  & 0.105 (152.6)&  & & 100 & $<$0.001 \\ \addlinespace

                                                    22& & & & 0.1 & {\bf\textcolor{red}{0.045}} (9.0)  & 0.047 (138.8) &  & & 100 & $<$0.001  \\
                                                    23& & & 15 & 0.5 &  {\bf\textcolor{red}{0.057}} (9.1)  & 0.061 (140.8) &  & & 100 &  $<$0.001  \\
                                                    24& & & & 2 &  {\bf\textcolor{red}{0.100}} (9.2) & 0.106 (151.5) &  & & 100 &$<$0.001  \\ \addlinespace
\hline
\end{tabular}
\label{tab4}
\normalsize
\label{table1}
\end{table}

Figure \ref{res3} illustrates how the type and number of basis functions used in the fitted OKFD models affect the prediction performance (minimum MSPE over the three trace-semivariogram models, averaged over the 100 realisations) for cases 21 and 22 considering medium sample sizes. Case 21 corresponds to simulated data generated by 7 B-splines with weak spatial dependence, whereas case 22 corresponds to simulated data generated by 15 B-splines with strong spatial dependence. In contrast to the simulated stationary Sp.T. models (cases 1-18) where prediction performance typically increases with the number of basis functions used in the fitted OKFD models, here we observe this phenomena only when Fourier basis functions are used in the fitted OKFD models. For B-splines, the best prediction performance is (naturally) achieved using the same number of B-splines in the OKFD fitted models as used to generate the simulated data set (7 for case 21 and 15 for case 22). In fact, using too many B-splines may give substantially poorer predictions, especially when the spatial dependence is weak, as for case 21, cf. Figure \ref{res3}. It can also be noted that the best OKFD model using B-splines has significantly smaller MSPE than the best OKFD model using Fourier basis functions. If the simulated data sets would have been generated by a set of Fourier basis functions instead, we would most likely see the opposite behaviour, i.e. that the same Fourier basis functions in the fitted OKFD model as in the data generation model  probably would give the best prediction performance, and do better than the OKFD models using B-splines. 

For the Sp.T. separable kriging models, it turned out (regardless of sample size) that it was advantageous to use universal kriging (estimating a deterministic time trend), especially for the cases with weak spatial dependence, whereas the prediction performance was about the same for cases with strong spatial dependence (Figure \ref{res3}). For the Sp.T. metric model, we observed the opposite behaviour, i.e., for cases with weak spatial dependence it was more advantageous to use ordinary kriging instead of universal kriging. 
\begin{figure}[!htb]
\begin{center}
\includegraphics[width=\textwidth]{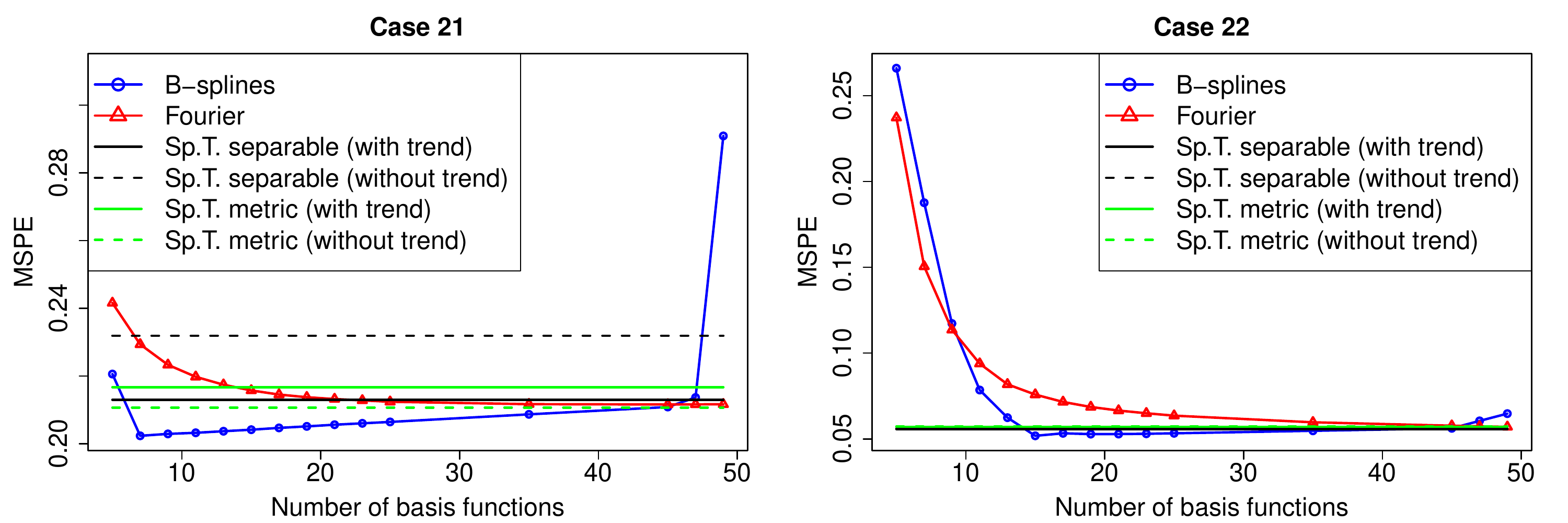} 
\caption{Prediction performance (minimum MSPE over the three trace-semivariogram models, averaged over the 100 realisations) for cases 21 and 22 considering medium sample sizes when the estimated OKFD model is based on different numbers ($p$) of basis functions, being both Fourier and cubic B-spline bases. The solid and dashed black lines represent the corresponding overall MSPE of the Sp.T. separable model with and without an estimated deterministic time trend, respectively. The solid and dashed green lines represent the corresponding overall MSPE of the Sp.T. metric model with and without an estimated deterministic time trend, respectively.}
\label{res3}
\end{center}
\end{figure}

\section{Applications}
\label{applications_main}
In this section we compare the prediction performance of the OKFD and the Sp.T. kriging models for two different data sets. The first data set consists of temperature curves recorded in the Maritimes Provinces of Canada, and the second corresponds to salinity curves obtained from the Caribbean coast of Colombia. 
\subsection{Spatial prediction of temperature curves in the Maritime Provinces of Canada}
\label{applications}
Here we analyse a meteorological data set, available in the R package geofd \citep{Giraldo2012geofd}. The data consists of temperature measurements recorded at $n=35$ weather stations at Canada's Atlantic coast in the Maritime Provinces (Figure \ref{figur1}, top panel). At each station, the daily mean temperature averaged over the period 1960-1994 (February 29th combined with February 28th) has been recorded. The resulting functional data are displayed in Figure \ref{figur1} (bottom panel), connected by light grey lines. 
\begin{figure}[!htb]
\begin{center}
\includegraphics[width=\textwidth]{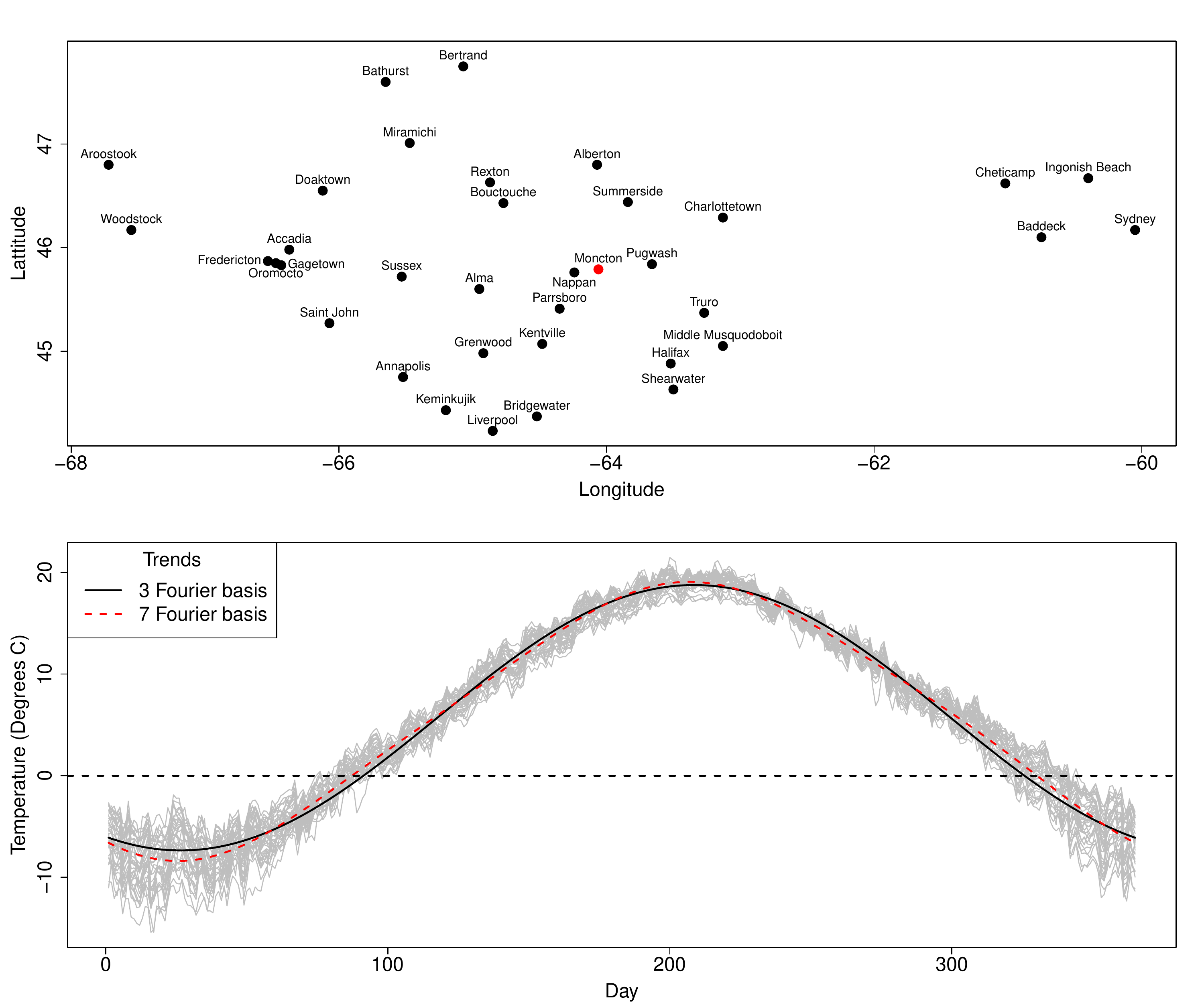} 
\caption{Locations of the  36 weather stations in the Canadian Maritime provinces (top panel) where the average (over 30 years) daily temperature curves (bottom panel) were registered. The bottom panel also presents the estimated common time trend specified as linear combinations of the first 3  and 7 Fourier basis functions, respectively.}
\label{figur1}
\end{center}
\end{figure} 
Using the R-package geofd, the data was first predicted by the OKFD model, 
which was estimated using 51, 101, 151, 201, 251, 301 and 351 Fourier basis functions. 
Three semivariogram models (exponential, spherical and stable) were fitted to the empirical trace-semivariogram by the OLS method. Thus, in total we estimated $7 \times 3=21$  OKFD models. Predictions were then made and evaluated by FCV in terms of their MSPEs \eqref{fcv}. 

The best prediction performance was achieved using the stable trace-semivariogram (Figure \ref{figur1variogramsF}, left panel) for all considered numbers of Fourier bases. Figure \ref{figur1variogramsF} (right panel) clearly reveals that the prediction error (minimum MSPE over the three trace-semivariogram models) decreases with the number of Fourier basis functions used in the fitted OKFD models. Thus, the best performance was attained with 351 Fourier basis functions and its MSPE was 0.5738. The average computational time for an estimated OKFD model based on 51 and 351 Fourier basis functions was less than one and three seconds, respectively.
\begin{figure}[!htb]
\begin{center}
\includegraphics[width=\textwidth]{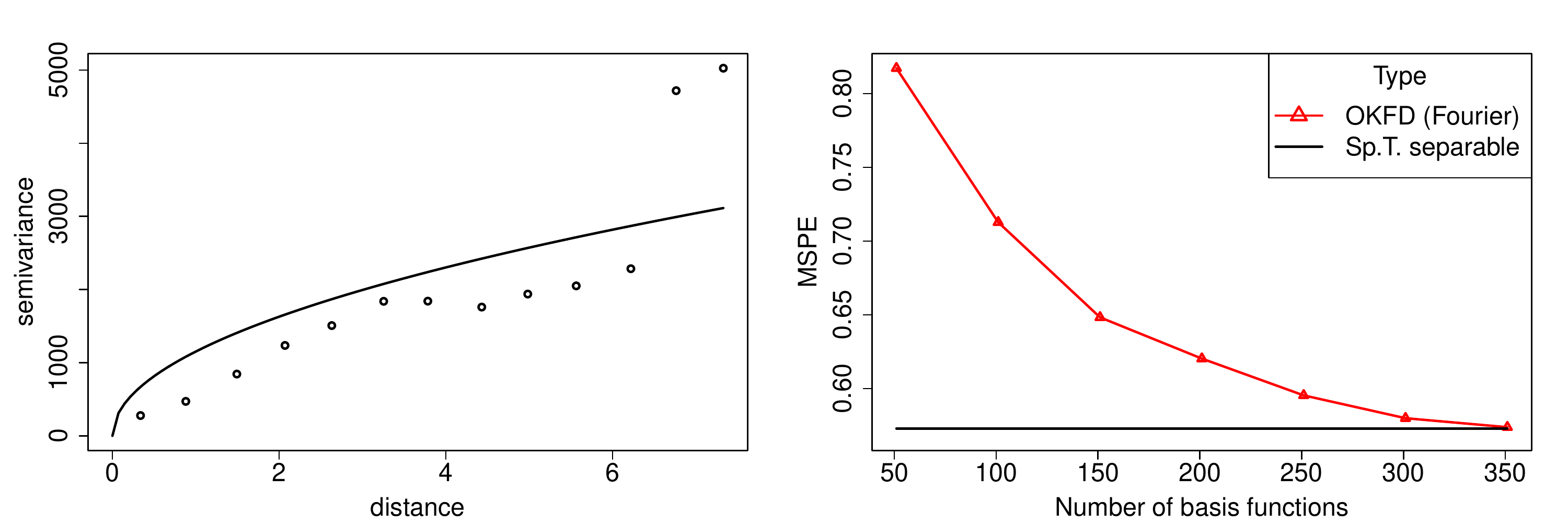}
\caption{Left panel: Empirical trace-semivariogram and the best fitted stable model for the Canadian temperature curves, represented by 351 Fourier basis functions. Right panel: Minimum MSPE over the three trace-semivariogram models for OKFD, based on different numbers of Fourier basis functions. The solid black line represents the MSPE of the best Sp.T. model.}
\label{figur1variogramsF}
\end{center}
\end{figure} 

The data was further predicted using Sp.T. kriging. Since the data show a clear time trend, universal Sp.T. kriging was first applied. The deterministic time trend was modelled by a linear combination of the 3 (and 7) first Fourier basis functions,
and estimated by the OLS method. The dependence structure of the resulting residuals was then estimated by fitting Sp.T. second-order stationary and isotropic semivariogram models to the empirical Sp.T. semivariogram of the residuals. The Sp.T. semivariogram models (separable, product-sum and metric) described in Section \ref{sptk} were estimated, letting their corresponding spatial, temporal and joint semivariogram models be altered between the exponential, spherical and stable semivariogram models. This resulted in 9 separable, 9 product-sum and 3 metric Sp.T. semivariogram models. As a comparison we also predicted the original data by Sp.T. ordinary kriging, using the same Sp.T. semivariogram models as for the universal Sp.T. kriging models. Thus, in total we investigated $(9+9+3) \times 3= 63$  Sp.T. kriging models. All  models were fitted to the data and predictions evaluated by FCV.  

Table \ref{table2Mari} presents the best (smallest MSPE) Sp.T. models, within each of the three groups of dependence structure (separable, product-sum and metric), with and without an estimated trend. The numbers in brackets report the corresponding average computational time in seconds over the estimated models. Many of the Sp.T. models have about the same prediction performance, with the exceptions of the Sp.T. metric models with estimated trend, which worked less well. The best Sp.T. models have approximately the same magnitude of MSPE as the best OKFD model (MSPE being 0.5738), but in terms of computational time, an OKFD model (taking 1-3 seconds to compute) was 100-10000 times faster to compute compared to a Sp.T. kriging model. 
\begin{table}[!htb]
\small
\caption{Prediction performance of different Sp.T. kriging models for the Canadian weather data. For each type of trend and Sp.T. variogram model, the (minimum) MSPE is reported. The numbers in parentheses represent the average computational time in seconds over the corresponding estimated models.} 
\centering
\begin{tabular}{c|ccc}
  \midrule
& \multicolumn{3}{c}{ MSPE}\\
 Trend & Sp.T. Separable& Sp.T. Product-sum & Sp.T. Metric  \\
 \midrule \addlinespace
 						No trend& 0.5730 ($1.8\cdot 10^2$) & 0.5861 ($1.3\cdot 10^4$) & 0.5730 ($1.3\cdot 10^4$)\\  \addlinespace
                                                  3 Fourier basis& 0.5730 ($1.8\cdot 10^2$) & 0.5731 ($1.3\cdot 10^4$) & 1.1126 ($1.4\cdot 10^4$)\\
                                                   7 Fourier basis& 0.5734 ($1.6\cdot 10^2$) & 0.5731 ($1.3\cdot 10^4$) & 1.0670 ($1.4\cdot 10^4$)\\ \addlinespace                                                     
\midrule
\end{tabular}
\normalsize
\label{table2Mari}
\end{table}

Figure \ref{figure_predictions} presents the observed daily temperatures at locations Bertrand (the location with the largest prediction error) and Moncton, together with the corresponding predicted values using the best OKFD and Sp.T. kriging models. It emphasizes that there are very small differences between the best OKFD and Sp.T. models (in terms of prediction performance).

\begin{figure}[!htb]
\begin{center}
\includegraphics[width=\textwidth]{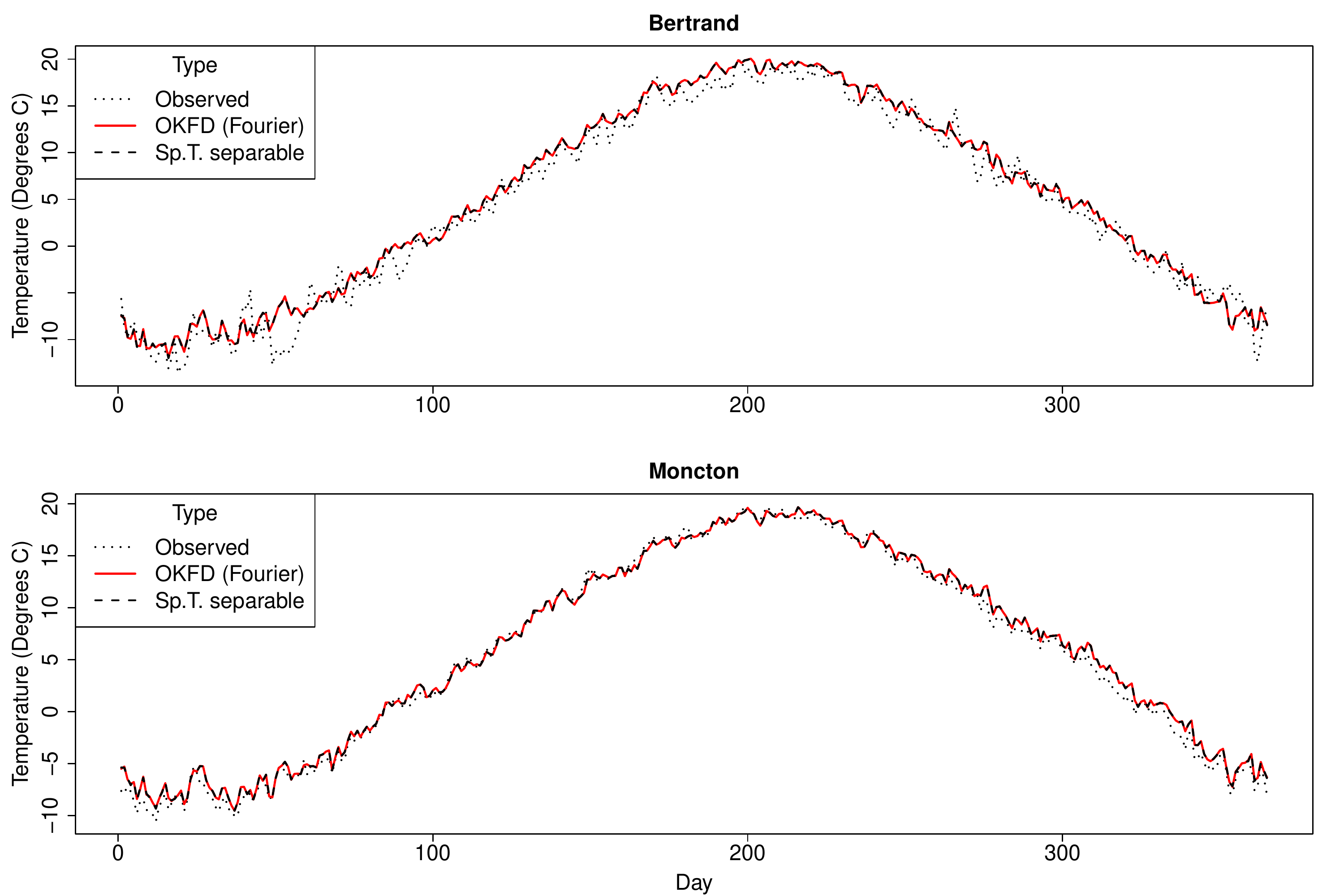}
\caption{Predicted temperatures at locations Bertrand (top) and Moncton (bottom) obtained by the best OKFD model (solid grey line) and the best Sp.T. model (dashed black line) together with the observed (dotted) values.}
\label{figure_predictions}
\end{center}
\end{figure}

This data set has previously been analysed by e.g. \cite{giraldo2009geostatistical} with the objective to demonstrate and compare the functional kriging methods OKFD, PWFK and FKTM. \cite{giraldo2009geostatistical} concluded that the three methods have similar FCV prediction performance when the  first 65 Fourier basis functions are used in \eqref{eqn:Xbasisexp} to represent the $\chi_{s_i}(t)$'s. In \cite{menafoglio2013universal} this data set was used to investigate the effect of using universal kriging for functional data (UKFD) instead of OKFD, also by representing the functional data with the first 65 Fourier basis functions. They concluded that UKFD performed better in terms of FCV compared to OKFD. The FCV performance was there computed with respect to the fitted data, thus differing from ours, where raw data has been used. 


\subsection{Spatial prediction of salinity curves on the Caribbean coast of Colombia}

Here we analyse a data set consisting of salinity measurements recorded at 21 monitoring stations of the lagoonal-estuarine system comprised by Ci\'enaga Grande de Santa Marta (CGSM) and Complex of Pajarales (CP) located on the Caribbean coast of Colombia, see Figure \ref{salinity1} (top panel). The data for each station were recorded biweekly from October 1988 to March 1991 (connected by lines in Figure \ref{salinity1}, bottom left panel). 
This data set has previously been used by \cite{reyes2015residual} to illustrate, evaluate and compare the performance of the functional kriging approaches OKFD, PWFK and FKTM when applied on residual curves after estimating a deterministic trend, so called ROKFD, RPWFK and RFKTM, as well as directly applied on the data. They came to the conclusion that ROKFD was the best alternative for performing functional kriging prediction, although the difference (in prediction performance) to RPWFK and RFKTM were small. Here, we will redo the same analysis using OKFD and ROKFD, and add predictions made by  Sp.T. kriging models, for comparison purposes. 



\begin{figure}[!htb]
\begin{center}
\includegraphics[width=\textwidth]{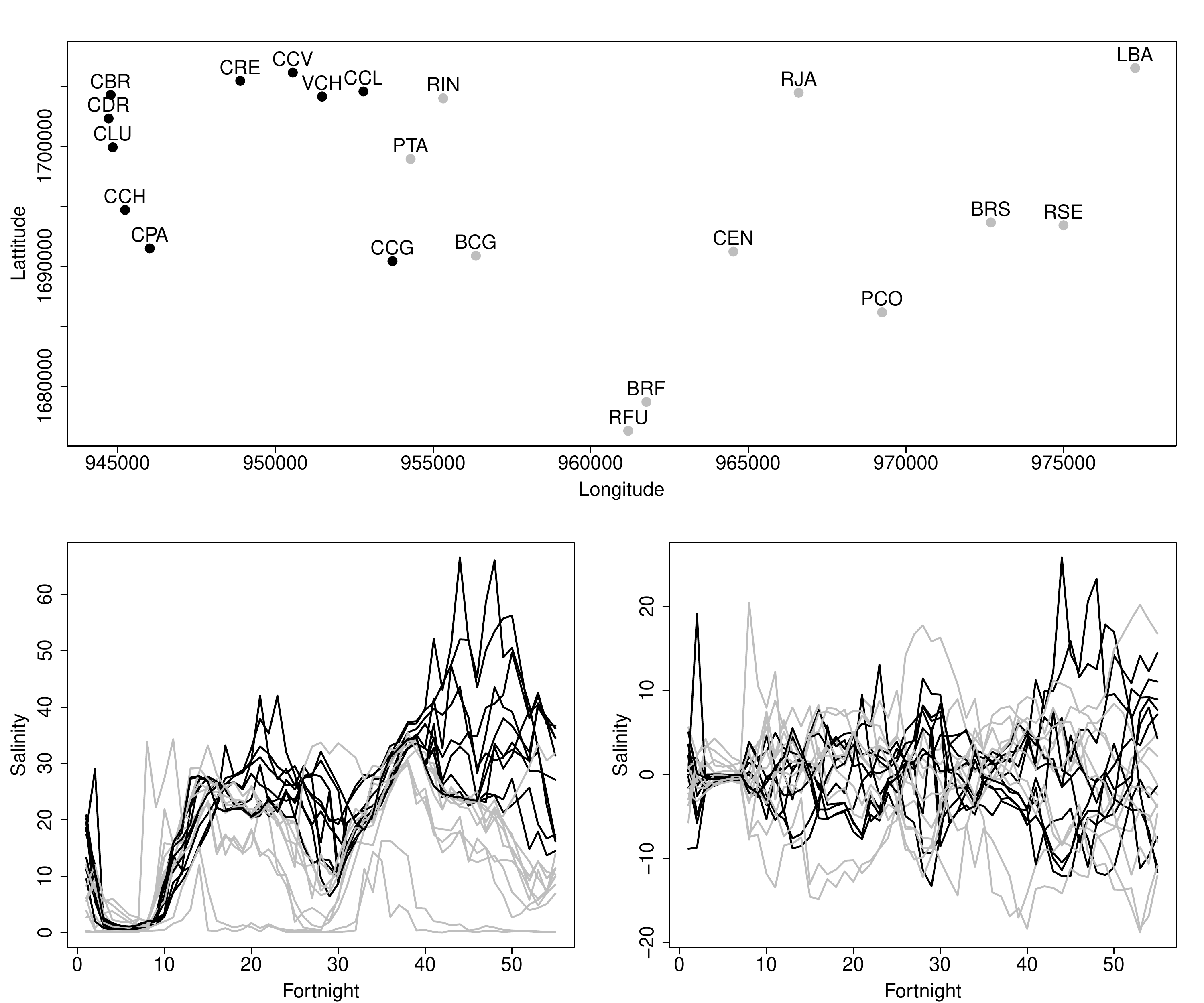} 
\caption{Top panel: The 21 monitoring stations of the lagoonal-estuarine system comprised by CGSM (stations: BCG, BRF, BRS, CEN, LBA, PCO, PTA, RFU, RIN, RJA, RSE) and CP (stations: CBR, CCG, CCH, CCL, CCV, CDR, CLU, CPA, CRE, VCH). Bottom left panel: Biweekly salinity data recorded on the 21 monitoring stations. Bottom right panel: Residuals obtained from the fitted functional regression model. The salinity data, residuals and the corresponding locations at CGSM and CP are represented by grey and black lines/points, respectively.}
\label{salinity1}
\end{center}
\end{figure} 


The salinity data is evidently not stationary, as there is a clear increasing trend from east to west, see Figure \ref{salinity1} (top and bottom left panel). For ROKFD, a deterministic trend was thus first estimated. We used the same trend model as \cite{reyes2015residual},
\begin{equation}
X_i(t)=\alpha(t)+\beta_1(t)\text{Longitude}_i+\beta_2(t)\text{Latitude}_i+\epsilon_i(t),
\label{salinitytrend}
\end{equation}   
where $X_i(t)$, $i=1,...,21$ are the salinity curves, and $\alpha(t)$, $\beta_1(t)$ and $\beta_2(t)$ are the functional parameters. 
For convenience, since we are evaluating prediction performance by FCV, we restricted the estimation of the functional parameters to the observed time points and thus fitted the model by OLS for each observed time point using the raw salinity data. \cite{reyes2015residual} fitted the trend based on smoothed salinity curves.
Once the trend was estimated, the resulting residual data $\epsilon_i(t_j)$, $i=1,...,21$, $j=1,...,55$ (connected by lines in Figure \ref{salinity1}, right panel) were formed. The Salinity data was predicted by ROKFD using the estimated deterministic trend combined with estimated OKFD models applied to the residual data. As a comparison we also estimated OKFD models directly on the original raw data, and used them to predict the Salinity data. 

In accordance with \cite{reyes2015residual}, we used B-splines basis functions to construct functional representations of both the original and the residual data. Specifically, we studied OKFD and ROKFD and their MSPEs using 5, 6, 7, 8, 9, 10, 15, 20, 30, 40 and 50 B-splines. Moreover, the exponential, spherical and stable semivariogram models were fitted to the empirical trace semivariograms (of the residual and original data) by OLS. Predictions were made and evaluated by FCV in terms of their MSPE for a total of $11 \times 3=33$ estimated OKFD models on the original data as well as 33 estimated ROKFD models by using the R-package geofd. 

The minimum MSPE, for each number of basis functions and trend used, was obtained by the stable trace-semivariogram (Figure \ref{salinity_fvariogram_mspe}, left panel). Figure \ref{salinity_fvariogram_mspe} (right panel) presents how the trend and the number of B-splines used in the fitted OKFD and ROKFD models affect the prediction performance (minimum MSPE over the three trace-variogram models). The prediction errors of OKFD decreases with the number of basis functions used while the performance of ROKFD approximately is the same, irrespective of the number of basis functions used (Figure \ref{salinity_fvariogram_mspe}). We believe that the performance of the latter is due to that a large part of the dependence structure is captured by the estimated Sp.T. deterministic trend. It is also noted that prediction based on ROKFD yields lower prediction errors compared to OKFD. 
Moreover, the computational time for all the 66 OKFD and ROKFD models was about the same, taking approximately 0.2 seconds each.


\begin{figure}[!htb]
\begin{center}
\includegraphics[width=\textwidth]{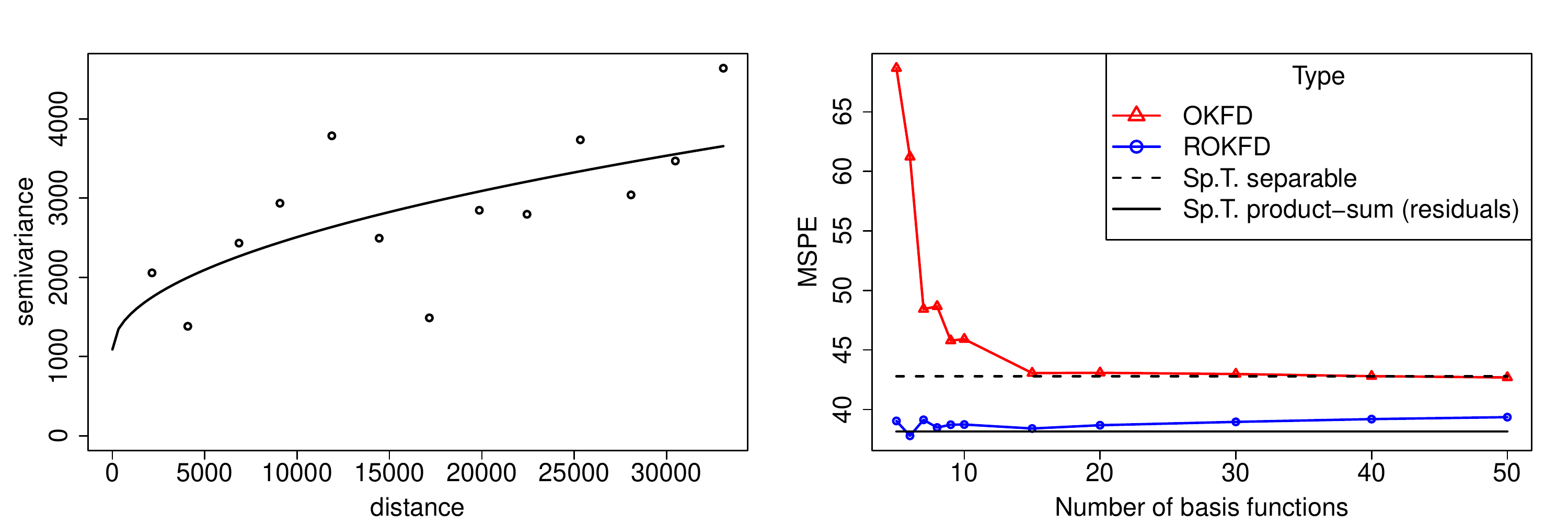} 
\caption{Left panel: Empirical trace semivariogram and the best fitted stable model for residual curves represented by 6 B-spline basis functions. Right panel: MSPEs for OKFD and ROKFD over different number of B-splines and trends used. The solid and dotted black line represents the MSPE of the best Sp.T. model with a Sp.T. product-sum and separable semivariogram model applied on the residuals and the original data, respectively.}
\label{salinity_fvariogram_mspe}
\end{center}
\end{figure} 

Sp.T. universal kriging models were also estimated and used to predict the Salinity data, based on the estimated deterministic trend specified in \eqref{salinitytrend}. To compare, we also applied Sp.T. ordinary kriging models to the original raw data.  The dependence structures were estimated by fitting the Sp.T. semivariograms (9 separable, 9 product-sum, 3 metric) to the empirical Sp.T. semivariograms (computed from both the residual and the original  data). Thus, in total $2\times (9+9+3)=42$ Sp.T. models were fitted to the data and then evaluated by FCV. 



\begin{table}[!htbp]
\small
\caption{Prediction performance of different Sp.T. kriging models for the salinity data. For each type of trend and Sp.T. variogram model, the (minimum) MSPE is reported. The numbers in parentheses represent the average computational time in seconds over the corresponding estimated models.} 
\centering
\begin{tabular}{Sc|ScScSc}
 \hline\hline
& \multicolumn{3}{c}{MSPE}\\
Trend & Separable& Product-sum & Metric  \\
\hline \addlinespace
No trend& 42.80 (17.5) & 42.85 (30.2) &43.80 (24.7)\\
Trend& 38.47 (19.4) & {\bf\textcolor{red}{38.15}} (31.9) &40.43 (25.8)\\  \addlinespace
\hline
\end{tabular}
\normalsize
\label{table2}
\end{table}

Table \ref{table2} presents the best Sp.T. models, in terms of minimum MSPE, within the three groups of dependence structure (separable, product-sum and metric) with and without deterministic trend. The numbers in brackets report the corresponding average computational times in seconds over the estimated models. The lowest MSPEs, being approximately of the same magnitude as those for the ROKFD models, were obtained by the Sp.T. separable and the product-sum universal kriging models (cf. Figure \ref{salinity_fvariogram_mspe}, right panel). It is also noted that the best Sp.T. ordinary kriging models, when applied on the original data, gave about the same size of the MSPEs as the best OKFD model. Figure \ref{figure_predictions_salinity} illustrates the predictions together with the observed salinity data at locations LBA and CCG (corresponding to the locations with the largest and smallest prediction errors, respectively) using the best ROKFD and Sp.T. kriging models. The predictions obtained by the two methods (ROKFD and the Sp.T. product-sum universal kriging) are very similar with the predictions in CCG performing good whereas the predictions in LBA (the farthest considered station in our data set, see Figure \ref{salinity1}, top panel) performing not as good. The computational times for the Sp.T. models (taking approximately 15-30 seconds per model, cf. Table \ref{table2}) are much higher than the ones for OKFD and ROKFD (taking approximately 0.2 seconds per model).

\begin{figure}[!htb]
\begin{center}
\includegraphics[width=\textwidth]{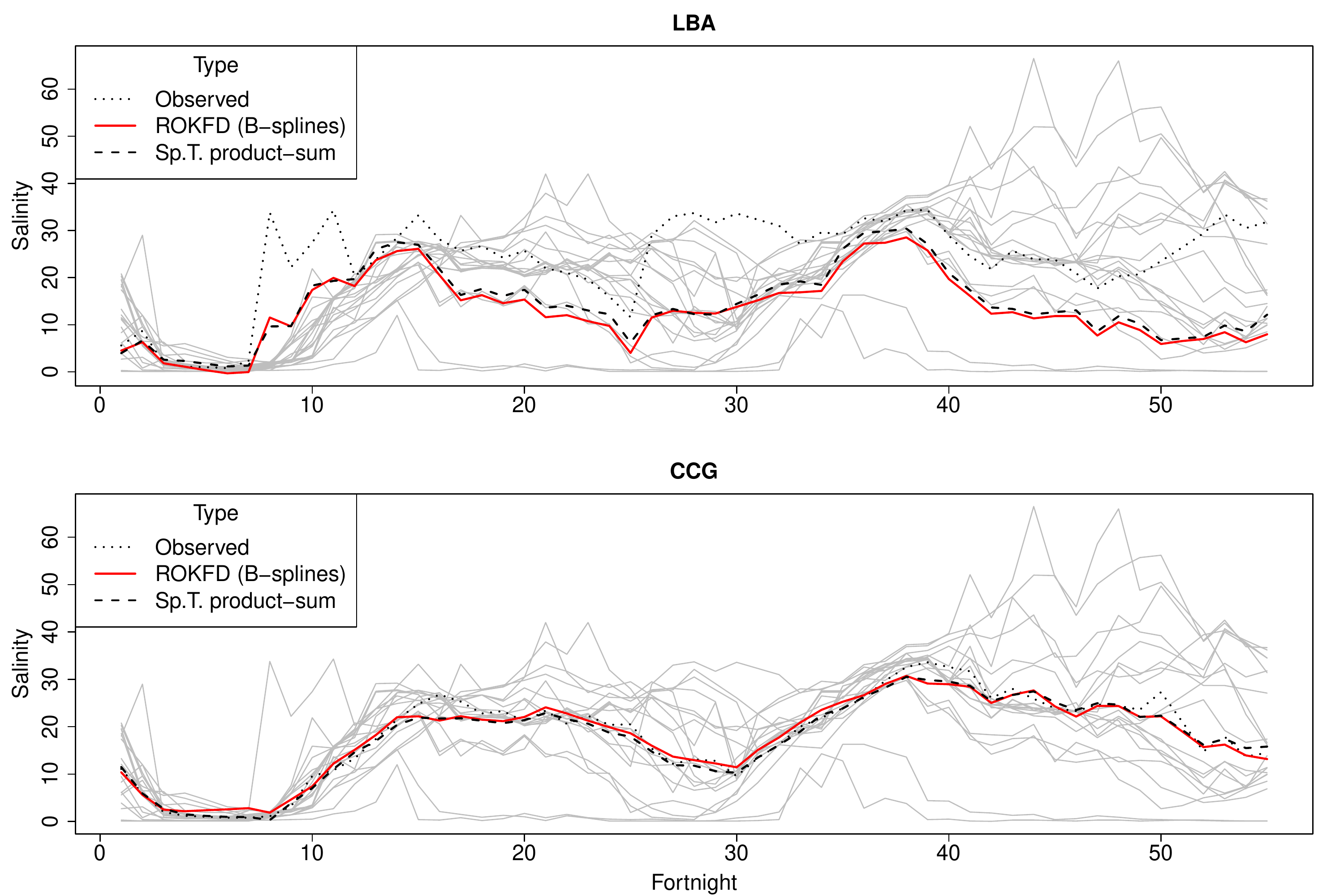}
\caption{Predicted temperatures at the locations LBA (top panel) and CCG (bottom panel) obtained by the best ROKFD model (solid grey line) and the best Sp.T. model (dashed black line) together with the observed (dotted) values.}
\label{figure_predictions_salinity}
\end{center}
\end{figure}

\section{Concluding remarks}
\label{conclusions}

In this paper we have presented and compared functional and Sp.T. kriging approaches to predict spatial functional random processes. Comparisons with respect to prediction performance and computational time has been performed, mainly through a simulation study and two real data sets. We restricted the comparison to Sp.T. kriging versus the functional kriging method OKFD, since the more flexible functional kriging approaches PWFK and FKTM coincide with OKFD in several situations (Sections \ref{secPWFK} and \ref{secFKTM}). 
Here we also contribute with new knowledge by proving that OKFD and PWFK coincide under certain conditions. 

Based on the simulation study and the analyses of the two data sets, we observed that the prediction performance (in terms of functional cross-validation) of OKFD normally was improved when the number of basis functions used to represent the functional data increased.
Furthermore, OKFD typically performed similarly or better than the Sp.T. kriging models for small and medium sample sizes. 
This is likely due to the more complex task of finding good estimates of the Sp.T. variogram compared to the trace-variograms used in OKFD, since trace-variograms have one dimension less. The large number of choices of Sp.T. variogram models and parameters to estimate makes the Sp.T. estimation process more vulnerable, especially for small data sets. 
For larger sample sizes, the Sp.T. kriging starts to perform better for the stationary Sp.T. processes, whereas OKFD continues to work best for the non-stationary Sp.T. (but stationary functional) processes. 
We also noted a clear tendency for OKFD to perform better relative to Sp.T. kriging, the stronger the temporal- and the weaker the spatial dependence considered.

For all considered cases, OKFD was computationally considerably faster than the Sp.T. kriging models.
The large matrices that need to be inverted in order to perform Sp.T. kriging prediction at each location, is the major reason for this fact.
One way to reduce the computational time for the Sp.T. kriging models could be to use only the local neighbourhood (e.g. the k closest neighbouring locations) when prediction is made. This can often be done without much loss in prediction performance.

The purpose of this study has been to shed light on the relative merits of functional and Sp.T. kriging methods for prediction of spatial functional random processes.
While functional kriging predicts complete curves on a given (time) domain, given observations on the same domain, the Sp.T. kriging methods make (a raster of) pointwise predictions of the curves and are not restricted to a given (time) domain. 
Experience from this study concludes that prediction performance of the two kriging approaches (functional and Sp.T.) in general is rather equal for stationary Sp.T. processes, with a tendency for functional kriging to work slightly better for small sample sizes and Sp.T. kriging to work slightly better for large sample sizes. For non-stationary Sp.T. processes, e.g. the presence of a common deterministic time trend and/or time varying variances and dependence structure, do not demand any extra modeling for functional kriging, whereas identification and modeling of trend and/or time varying dependence is necessary for Sp.T. kriging. From a modelers perspective, the Sp.T. kriging methods demands more work with a larger risk of choosing a suboptimal model. Moreover, from  a computational perspective functional kriging is substantially faster than Sp.T. kriging.

\newpage
\appendix

\section{Proof of Proposition \ref{prop}}
\label{pwfkEokfd}
Below we present a proof of Proposition \ref{prop}. The proof relies on Lemma \ref{induction_lemma}, which we first state and prove.
\begin{lem}
\label{induction_lemma}
Assume that we have a symmetric matrix of the form
\begin{equation*}
 \mathbf{Q}=\begin{pmatrix}
 \mathbf{W} & \mathbf{W}-c_{12}\mathbf{G} & \cdots & \mathbf{W}-c_{1n}\mathbf{G} &\mathbf{I} \\
  \mathbf{W}-c_{21}\mathbf{G} & \mathbf{W}  & \cdots & \mathbf{W}-c_{2n}\mathbf{G}&\mathbf{I} \\
  \vdots  & \vdots  & \ddots & \vdots&\vdots  \\
    \mathbf{W}-c_{n1}\mathbf{G}  &  \mathbf{W}-c_{n2}\mathbf{G}  & \cdots & \mathbf{W} &\mathbf{I} \\
  \mathbf{I}&\mathbf{I}&\cdots&\mathbf{I}&\mathbf{0}
 \end{pmatrix}_{k(n+1)\times k(n+1)},\\
 \end{equation*}
where $\mathbf{W}$ and $\mathbf{G}$ are symmetric $k\times k$ matrices, $\mathbf{I}$ is a $k\times k$ identity matrix, and $c_{ij}=c_{ji}$ for $i\neq j,$ $i,j=1,...,n$ are constants. Given that the inverse of $\mathbf{G}$ exists, the inverse of $\mathbf{Q}$ satisfies
 \begin{equation}
 \label{inverse}
  \mathbf{Q}^{-1}=\begin{pmatrix}
 k_{11} \mathbf{G}^{-1} &  k_{12} \mathbf{G}^{-1}  & \cdots &  k_{1n} \mathbf{G}^{-1}  &k_{1}\mathbf{I} \\
 k_{21} \mathbf{G}^{-1} &  k_{22} \mathbf{G}^{-1}  & \cdots &  k_{2n} \mathbf{G}^{-1} &k_{2}\mathbf{I} \\
  \vdots  & \vdots  & \ddots & \vdots&\vdots  \\
 k_{n1} \mathbf{G}^{-1}   &  k_{n2} \mathbf{G}^{-1}  & \cdots &  k_{nn} \mathbf{G}^{-1} &k_{n}\mathbf{I} \\
  k_{1}\mathbf{I}&k_{2}\mathbf{I}&\cdots&k_{n}\mathbf{I}&c\mathbf{G}-\mathbf{W}
 \end{pmatrix}_{k(n+1)\times k(n+1)}, \\ 
 \end{equation}
where the $k_{ij}$'s, $k_i$'s, and $c$ are constants, determined by the $c_{ij}$'s, such that $k_{ij}=k_{ji}$ for $i,j=1,...,n$, $\sum_{i=1}^n k_{ij}(=\sum_{j=1}^n k_{ij})=0$ for all $j$'s (and $i$'s), and $\sum_{i=1}^n k_i=1$. \\
\end{lem}
Note that the $k_{ij}$'s and the $k_i$'s also change with $n$. However, for notational simplicity we suppress the dependence of $n$ in the $k_{ij}$'s and the $k_i$'s.
\begin{proof}[Proof of Lemma \ref{induction_lemma}]
The proof of \eqref{inverse} is done by induction. The proof uses the following, equivalent, block matrix inversion formulas
   \begin{multline}
   \label{B2}
\begin{pmatrix}
  \mathbf{A} & \mathbf{B} \\
    \mathbf{C} & \mathbf{D} \\
 \end{pmatrix}^{-1} = \begin{pmatrix}
  \mathbf{A}^{-1}+ \mathbf{A}^{-1} \mathbf{B} (\mathbf{D}-\mathbf{C}\mathbf{A}^{-1}\mathbf{B})^{-1}\mathbf{C} \mathbf{A}^{-1}& -\mathbf{A}^{-1} \mathbf{B} (\mathbf{D}-\mathbf{C}\mathbf{A}^{-1}\mathbf{B})^{-1}\\
  -(\mathbf{D}-\mathbf{C}\mathbf{A}^{-1}\mathbf{B})^{-1}\mathbf{C} \mathbf{A}^{-1} &(\mathbf{D}-\mathbf{C}\mathbf{A}^{-1}\mathbf{B})^{-1} \\
 \end{pmatrix}=\\ \begin{pmatrix}
  (\mathbf{A}-\mathbf{B}\mathbf{D}^{-1}\mathbf{C})^{-1}& -(\mathbf{A}-\mathbf{B}\mathbf{D}^{-1}\mathbf{C})^{-1}\mathbf{B}\mathbf{D}^{-1}\\
 -\mathbf{D}^{-1}\mathbf{C} (\mathbf{A}-\mathbf{B}\mathbf{D}^{-1}\mathbf{C})^{-1} &\mathbf{D}^{-1}+\mathbf{D}^{-1}\mathbf{C} (\mathbf{A}-\mathbf{B}\mathbf{D}^{-1}\mathbf{C})^{-1}\mathbf{B}\mathbf{D}^{-1} \\
 \end{pmatrix}=\begin{pmatrix}
  \mathbf{E} & \mathbf{F} \\
    \mathbf{H} & \mathbf{K} \\
 \end{pmatrix},
 \bigskip
 \end{multline}
where $\mathbf{A}$ and $\mathbf{D}$ are square matrices allowed to be of different size.

For $n=2$ we let
\begin{equation}
  \mathbf{Q}=\left(
\begin{array}{c|cc}
  \label{inv2}
  \mathbf{W} & \mathbf{W}-c_{12}\mathbf{G} & \mathbf{I} \\ \hline
    \mathbf{W}-c_{21}\mathbf{G} & \mathbf{W} & \mathbf{I}\\
 \mathbf{I}  & \mathbf{I} & \mathbf{I} \\
\end{array}
\right)=\begin{pmatrix}
  \mathbf{A} & \mathbf{B} \\
    \mathbf{C} & \mathbf{D} \\
 \end{pmatrix}.
\end{equation}
By \eqref{B2} the inverse of $\mathbf{D}$ is
  \begin{equation}
  \label{inv1}
  \mathbf{D}^{-1}=\begin{pmatrix}
  \mathbf{W} & \mathbf{I} \\
    \mathbf{I} & \mathbf{0} \\
 \end{pmatrix}^{-1} =\begin{pmatrix}
  \mathbf{0} & \mathbf{I} \\
    \mathbf{I} & -\mathbf{W} \\
 \end{pmatrix}.
 \bigskip
 \end{equation}
From \eqref{B2} and \eqref{inv1} it is straight forward to verify that the inverse of \eqref{inv2} can be expressed as 
  \begin{equation*}
\mathbf{Q}^{-1}=\begin{pmatrix} \frac{1}{2c_{12}}\mathbf{G}^{-1}&-\frac{1}{2c_{12}}\mathbf{G}^{-1}&\frac{1}{2}\mathbf{I}\\ \\
-\frac{1}{2c_{12}}\mathbf{G}^{-1}&\frac{1}{2c_{12}}\mathbf{G}^{-1}&\frac{1}{2}\mathbf{I}\\ \\
\frac{1}{2}\mathbf{I}&\frac{1}{2}\mathbf{I}&\frac{c_{12}}{2}\mathbf{G}-\mathbf{W}\\
 \end{pmatrix}=\begin{pmatrix} k_{11}\mathbf{G}^{-1}&k_{12}\mathbf{G}^{-1}&k_{1}\mathbf{I}\\ \\
k_{21}\mathbf{G}^{-1}&k_{22}\mathbf{G}^{-1}&k_{2}\mathbf{I}\\ \\
k_{1}\mathbf{I}&k_{2}\mathbf{I}&c\mathbf{G}-\mathbf{W}\\
 \end{pmatrix},
\end{equation*}
where $k_{12}=k_{21}$, $\sum_{i=1}^2 k_{ij}=1/2c_{12}-1/2c_{12}=0$, and $\sum_{i=1}^2k_i=1/2+1/2=1$. Thus, \eqref{inverse} holds for $n=2$.
 
Assume now that \eqref{inverse} holds for $n=m-1$. The proof is complete if we can show that \eqref{inverse} holds for $n=m$. For $n=m$ we split the matrix $\mathbf{Q}$ in four blocks
 \begin{equation*}
  \mathbf{Q}=\left(
\begin{array}{c|cccc}
 \mathbf{W} & \mathbf{W}-c_{12}\mathbf{G} & \cdots & \mathbf{W}-c_{1m}\mathbf{G} &\mathbf{I} \\ \hline
  \mathbf{W}-c_{21}\mathbf{G} & \mathbf{W}  & \cdots & \mathbf{W}-c_{2m}\mathbf{G}&\mathbf{I} \\
  \vdots  & \vdots  & \ddots & \vdots&\vdots  \\
    \mathbf{W}-c_{m1}\mathbf{G}  &  \mathbf{W}-c_{m2}\mathbf{G}  & \cdots & \mathbf{W} &\mathbf{I} \\
  \mathbf{I}&\mathbf{I}&\cdots&\mathbf{I}&\mathbf{0}
\end{array}
\right)_{k(m+1)\times k(m+1)}=\begin{pmatrix}
  \mathbf{A} & \mathbf{B} \\
    \mathbf{C} & \mathbf{D} \\
 \end{pmatrix}.
\end{equation*}
 By assumption the inverse of $\mathbf{D}$ is of the form $\eqref{inverse}$. Using this fact we have that 
 \begin{multline}
 \label{BD}
 \mathbf{B}\mathbf{D}^{-1}=\Bigg( \sum_{i=1}^{m-1} \Big(\mathbf{W}-c_{1(i+1)}\mathbf{G}\Big)k_{i1}\mathbf{G}^{-1}+k_1\mathbf{I},...,\sum_{i=1}^{m-1} \Big(\mathbf{W}-c_{1(i+1)}\mathbf{G}\Big)k_{i(m-1)}\mathbf{G}^{-1}+k_{m-1}\mathbf{I}, \\ \sum_{i=1}^{m-1} \Big(\mathbf{W}-c_{1(i+1)}\mathbf{G}\Big)k_i+c\mathbf{G}-\mathbf{W}\Bigg)=\\ \Bigg( \Big(k_1-\sum_{i=1}^{m-1}c_{1(i+1)}k_{i1}\Big)\mathbf{I},..., \Big(k_{m-1}-\sum_{i=1}^{m-1}c_{1(i+1)}k_{i(m-1)}\Big)\mathbf{I}, \Big(c-\sum_{i=1}^{m-1}c_{1(i+1)}k_i\Big)\mathbf{G} \Bigg),
 \end{multline}
 where the second equality uses the induction hypothesis that $\sum_{i=1}^{m-1} k_{ij}=0$ for all $j=1,...,m-1$, and $\sum_{i=1}^{m-1}k_i=1$. Further, it is relative straight forward, using the induction hypothesis, to verify that
 \begin{multline*}
 \mathbf{A}-\mathbf{B}\mathbf{D}^{-1}\mathbf{C}=\mathbf{W}-\Bigg(\sum_{j=1}^{m-1}\Big(k_j-\sum_{i=1}^{m-1}c_{1(i+1)}k_{ij}\Big)\Big(\mathbf{W}-c_{1(j+1)}\mathbf{G}\Big)+\Big(c-\sum_{i=1}^{m-1}c_{1(i+1)}k_i\Big)\mathbf{G}\Bigg)=\\  \Bigg(2\sum_{i=1}^{m-1}c_{1(i+1)}k_i- \sum_{i=1}^{m-1}\sum_{j=1}^{m-1} c_{1(i+1)}c_{1(j+1)}k_{ij}-c \Bigg) \mathbf{G}=k^{*}\mathbf{G}.
 \end{multline*}
Hence, 
 \begin{equation}
 \label{E}
 \mathbf{E}=(\mathbf{A}-\mathbf{B}\mathbf{D}^{-1}\mathbf{C})^{-1}=\frac{1}{k^{*}}\mathbf{G}^{-1}=k_{11}^{*}\mathbf{G}^{-1}.\\
 \end{equation}
Combining \eqref{BD} and \eqref{E}, we obtain
 \begin{multline*}
 \mathbf{F}=-\mathbf{E}\mathbf{B}\mathbf{D}^{-1}=\\ -\frac{1}{k^{*}}\Bigg( \Big(k_1-\sum_{i=1}^{m-1}c_{1(i+1)}k_{i1}\Big)\mathbf{G}^{-1},..., \Big(k_{m-1}-\sum_{i=1}^{m-1}c_{1(i+1)}k_{i(m-1)}\Big)\mathbf{G}^{-1}, \Big(c-\sum_{i=1}^{m-1}c_{1(i+1)}k_i\Big)\mathbf{I} \Bigg)=\\ \Big( k_{12}^{*}\mathbf{G}^{-1},...,k_{1m}^{*}\mathbf{G}^{-1}, k_1^{*} \mathbf{I} \Big),
 \end{multline*}
and due to the fact that $\mathbf{D}^{-1}$ and $\mathbf{E}$ are symmetric (using that the inverse of a symmetric matrix is again symmetric) and $\mathbf{C}=\mathbf{B}^{\intercal}$ we also have that
  \begin{equation}
  \label{H}
 \mathbf{H}=-\mathbf{D}^{-1}\mathbf{C}\mathbf{E}=-(\mathbf{E}\mathbf{B}\mathbf{D}^{-1})^{\intercal}=\mathbf{F}^{\intercal}.
 \end{equation} 
Moreover, by combining \eqref{BD}, \eqref{H} and the inverse of $\mathbf{D}$ it can also be shown that the $(i,j)$th $k\times k$ block matrix entry in the $km\times km$ matrix $\mathbf{K}=\mathbf{D}^{-1}-\mathbf{H}\mathbf{B}\mathbf{D}^{-1}=\{ \mathbf{K}_{ij} \}_{i,j=1,...,m}$ equals 
\begin{equation}
\label{klika}
\mathbf{K}_{ij}=\Bigg( k_{ij}+\frac{1}{k^{*}}\Big(k_i-\sum_{l=1}^{m-1}c_{1(l+1)}k_{li}\Big)\Big(k_j-\sum_{l=1}^{m-1}c_{1(l+1)}k_{lj}\Big)\Bigg) \mathbf{G}^{-1}=k_{(i+1)(j+1)}^{*}\mathbf{G}^{-1}, \hspace{5pt} i,j=1,...,m-1,
\end{equation}
\begin{equation*}
\mathbf{K}_{im}=\mathbf{K}_{mi}=\Bigg( k_{i}+\frac{1}{k^{*}}\Big(k_i-\sum_{l=1}^{m-1}c_{1(l+1)}k_{li}\Big)\Big(c-\sum_{l=1}^{m-1}c_{1(l+1)}k_{l}\Big)\Bigg) \mathbf{I}=k_{i+1}^{*}\mathbf{I}, \hspace{5pt} i=1,...,m-1,
\end{equation*}
and
 \begin{equation*}
\mathbf{K}_{mm}=c\mathbf{G}-\mathbf{W}+\frac{1}{k^{*}}\Big( c+\sum_{l=1}^{m-1}c_{1(l+1)}k_{l}\Big)^2 \mathbf{G}=c^{*}\mathbf{G}-\mathbf{W}.
\end{equation*}
Thus, the inverse $\mathbf{Q}^{-1}=\begin{pmatrix}
  \mathbf{E} & \mathbf{F} \\
    \mathbf{H} & \mathbf{K} \\
 \end{pmatrix}$ is of the same form as in $\eqref{inverse}$. It now remains to show that the $k^{*}$-coefficients satisfy the same conditions as the $k$'s in Lemma \ref{induction_lemma}. From \eqref{E} and \eqref{klika} and the fact that $k_{ij}=k_{ji}$, $i,j=1,...,m-1$, and that $\mathbf{H}=\mathbf{F}^{\intercal}$, we have that $k^{*}_{ij}=k^{*}_{ji}$, $i,j=1,...,m$. Moreover, by the induction hypothesis, for $j=1$ in $k^{*}_{ij}$, we obtain
 \begin{equation*}
 \sum_{i=1}^{m}k_{i1}^{*}=\frac{1}{k^{*}}-\frac{1}{k^{*}}\sum_{i=2}^{m}\Big( k_{i-1}-\sum_{l=1}^{m-1}c_{1(l+1)}k_{l(i-1)}\Big)=\frac{1}{k^{*}}-\frac{1}{k^{*}}=0,
 \end{equation*}
 and for $j=2,...,m$ we have
 \begin{multline*}
 \sum_{i=1}^{m}k_{ij}^{*}= -\frac{1}{k^{*}}\Big( k_{j-1}-\sum_{l=1}^{m-1}c_{1(l+1)}k_{l(j-1)}\Big)\\+\sum_{i=2}^{m}\Bigg( k_{(i-1)(j-1)}+\frac{1}{k^{*}}\Big(k_{i-1}-\sum_{l=1}^{m-1}c_{1(l+1)}k_{l(i-1)}\Big)\Big(k_{j-1}-\sum_{l=1}^{m-1}c_{1(l+1)}k_{l(j-1)}\Big)\Bigg)=\\ -\frac{1}{k^{*}}\Big( k_{j-1}-\sum_{l=1}^{m-1}c_{1(l+1)}k_{l(j-1)}\Big)+\frac{1}{k^{*}}\Big( k_{j-1}-\sum_{l=1}^{m-1}c_{1(l+1)}k_{l(j-1)}\Big)=0.
 \end{multline*}
 Also,
 \begin{multline*}
 \sum_{i=1}^{m}k_i^{*}=-\frac{1}{k^{*}}\Big( c-\sum_{l=1}^{m-1}c_{1(l+1)}k_l\Big)+\sum_{i=2}^{m}\Bigg( k_{i-1}+\frac{1}{k^{*}}\Big(k_{i-1}-\sum_{l=1}^{m-1}c_{1(l+1)}k_{l(i-1)}\Big)\Big(c-\sum_{l=1}^{m-1}c_{1(l+1)}k_{l}\Big)\Bigg)=\\ -\frac{1}{k^{*}}\Big( c-\sum_{l=1}^{m-1}c_{1(l+1)}k_l \Big)+1+\frac{1}{k^{*}}\Big( c-\sum_{l=1}^{m-1}c_{1(l+1)}k_l\Big)=1.
 \end{multline*}
Hence, formula \eqref{inverse} holds for $n=m$, and by the induction principle we have that \eqref{inverse} is true for all integers $n$ larger than 1. We have thus proved Lemma \ref{induction_lemma}.
\end{proof}

\begin{proof}[Proof of Proposition \ref{prop}]
Under the assumption in Proposition \ref{prop} we here show that the coefficients of the functional kriging weights \eqref{eqn:Lbasisexp} of PWFK, which are obtained by minmising \eqref{eqn:mise} subject to the unbiasedness constraint of the predictor ($\sum_{i=1}^n \lambda_i(t)=1$, for all $t\in T$), yields weights that are constant over time, i.e., $\lambda_i(t)=\mathbf{b}_i^{\intercal} \mathbf{B}_{\lambda}(t)=\lambda_i$, $i=1,...,n$. \cite{giraldo2010continuous} showed that the solution of the optimisation problem is given by the solution of the system $\mathbf{Q} \boldsymbol{\beta}=\mathbf{J}    \implies   \hat{\boldsymbol{\beta}}=\mathbf{Q}^{-1}\mathbf{J}$, where
\begin{equation*}
\mathbf{Q} = 
 \begin{pmatrix}
  \mathbf{Q}_1 & \mathbf{Q}_{12} & \cdots & \mathbf{Q}_{1n}&\mathbf{I} \\
  \mathbf{Q}_{21} & \mathbf{Q}_{2} & \cdots & \mathbf{Q}_{2n}&\mathbf{I} \\
  \vdots  & \vdots  & \ddots & \vdots&\vdots  \\
   \mathbf{Q}_{n1} & \mathbf{Q}_{n2} & \cdots & \mathbf{Q}_{n}&\mathbf{I} \\
  \mathbf{I}&\mathbf{I}&\cdots&\mathbf{I}&\mathbf{0}
 \end{pmatrix}, \hspace{20pt} 
 \boldsymbol{\beta}=\begin{pmatrix} \mathbf{b}_1 \\ \mathbf{b}_2\\ \vdots \\ \mathbf{b}_n\\ \mathbf{m} \end{pmatrix}, \hspace{10pt} \text{and} \hspace{10pt} 
 \mathbf{J}=\begin{pmatrix} \mathbf{J_1}\\ \mathbf{J_2}\\ \vdots \\ \mathbf{J_n}\\ \mathbf{c} \end{pmatrix},
 \end{equation*}
 where 
  \begin{equation}
 \mathbf{Q}_i=\int_{T} \mathbf{B}_{\lambda}(t)\mathbf{B}^{\intercal}(t)V(\mathbf{a}_i)\mathbf{B}(t)\mathbf{B}_{\lambda}^{\intercal}(t)dt=\int_{T} \sigma^2_i(t)\mathbf{B}_{\lambda}(t)\mathbf{B}_{\lambda}^{\intercal}(t)dt,
 \label{f1}
 \end{equation}
 \begin{equation}
 \mathbf{Q}_{ij}=\int_{T} \mathbf{B}_{\lambda}(t)\mathbf{B}^{\intercal}(t)C(\mathbf{a}_i,\mathbf{a}_j)\mathbf{B}^{}(t)\mathbf{B}_{\lambda}^{\intercal}(t)dt=\int_{T} \sigma^2_{ij}(t)\mathbf{B}_{\lambda}(t)\mathbf{B}_{\lambda}^{\intercal}(t)dt,
 \label{f2}
 \end{equation}
  \begin{equation}
 \mathbf{J}_i=\int_{T} \mathbf{B}_{\lambda}(t)\mathbf{B}^{\intercal}(t)C(\mathbf{a}_0,\mathbf{a}_i)\mathbf{B}(t)dt=\int_{T} \sigma^2_{0i}(t)\mathbf{B}_{\lambda}(t)dt,
  \label{f3}
 \end{equation}
 $\mathbf{m}^{\intercal}=(m_1,...,m_K)$ are the K Lagrangian multipliers and $\mathbf{c}$ is an unbiasedness constraint vector satisfying $\mathbf{c}^{\intercal}\mathbf{B}_{\lambda}(t)$=1.
 
By the assumption we have that $\mathbf{a}_i=\mathbf{P}\mathbf{r}_i$ with $V(\mathbf{r}_i)=\mathbf{D}(0)=\sigma^2\mathbf{I}$ and $C(\mathbf{r}_i,\mathbf{r}_j)=\mathbf{D}(h_{ij})=(\sigma^2-\gamma(h_{ij}))\mathbf{I}$, where 
$\gamma(h)$ is the common semivariogram function depending on the distance $h_{ij}=h_{ji}=\|s_i-s_j\|$ between two locations $s_i$ and $s_j$. Therefore, $\sigma_i(t)$ and $\sigma_{ij}(t)$ in expressions \eqref{f1}, \eqref{f2} and \eqref{f3} equals 
\begin{equation*}
\sigma^2_i(t)=\mathbf{B}^{\intercal}(t)V(\mathbf{a}_i)\mathbf{B}^{}(t)=\mathbf{B}^{\intercal}(t)\mathbf{P}\mathbf{D}(0)\mathbf{P}^{\intercal}\mathbf{B}(t)=\sigma^2(t), \hspace{10pt} i=1,...,n,
\end{equation*}
and
\begin{equation}
\label{sigmaij}
\sigma_{ij}(t)= \mathbf{B}^{\intercal}(t)C(\mathbf{a}_i,\mathbf{a}_j)\mathbf{B}(t)=       \mathbf{B}^{\intercal}(t)\mathbf{P}\mathbf{D}(h)\mathbf{P}^{\intercal}\mathbf{B}(t) = 
\sigma^2(t)-\gamma(h_{ij})f(t) \hspace{10pt} i<j, i,j=0,1,...,n,
\end{equation}
where $f(t)= \mathbf{B}^{\intercal}(t)\mathbf{P}\mathbf{I}\mathbf{P}^{\intercal}\mathbf{B}(t)$. Thus, $\mathbf{Q}_i=\mathbf{W}$ and $\mathbf{Q}_{ij}=\mathbf{W}-\gamma(h_{ij})\mathbf{G}$ for $i,j=1,...,n$, where $\mathbf{W}=\int_{T} \sigma^2(t)\mathbf{B}_{\lambda}(t)\mathbf{B}_{\lambda}^{\intercal}(t)dt$ and $\mathbf{G}=\int_{T} f(t)\mathbf{B}_{\lambda}(t)\mathbf{B}_{\lambda}^{\intercal}(t)dt$. Hence, the system we want to solve $\mathbf{Q}\boldsymbol{\beta}=\mathbf{J}$ may be expressed as
\begin{equation*}
 \begin{pmatrix}
 \mathbf{W} & \mathbf{W}-\gamma(h_{12})\mathbf{G} & \cdots & \mathbf{W}-\gamma(h_{1n})\mathbf{G} &\mathbf{I} \\
  \mathbf{W}-\gamma(h_{21})\mathbf{G} & \mathbf{W}  & \cdots & \mathbf{W}-\gamma(h_{2n})\mathbf{G}&\mathbf{I} \\
  \vdots  & \vdots  & \ddots & \vdots&\vdots  \\
    \mathbf{W}-\gamma(h_{n1})\mathbf{G}  &  \mathbf{W}-\gamma(h_{n2})\mathbf{G}  & \cdots & \mathbf{W} &\mathbf{I} \\
  \mathbf{I}&\mathbf{I}&\cdots&\mathbf{I}&\mathbf{0}
 \end{pmatrix}
 \begin{pmatrix} \mathbf{b}_1 \\ \mathbf{b}_2\\ \vdots \\ \mathbf{b}_n\\ \mathbf{m} \end{pmatrix}=
 \begin{pmatrix} \mathbf{J}_1\\ \mathbf{J}_2\\ \vdots \\ \mathbf{J}_n\\ \mathbf{c} \end{pmatrix}.
 \end{equation*}
 
  By Lemma \ref{induction_lemma} the inverse of $\mathbf{Q}$ is of the form \eqref{inverse} (as long as the inverse of $\mathbf{G}$ exists) and thus it follows that the solution of such system for any $\mathbf{b}_i$ is of the form 
\begin{equation*}
\mathbf{b}_i=\sum_{j=1}^n k_{ij}\mathbf{G}^{-1}\mathbf{J}_j+k_i\mathbf{c}, \hspace{10pt} i=1,...,n,
\label{exp1}
\end{equation*}
which can be rewritten as
\begin{equation}
\mathbf{G}(\mathbf{b}_i-k_i\mathbf{c})=\sum_{j=1}^n k_{ij}\mathbf{J}_j, \hspace{10pt} i=1,...,n,
\label{G}
\end{equation}
where the $k_{ij}$'s and $k_i$'s are constants determined by the $\gamma(h_{ij})$'s such that $k_{ij}=k_{ji}$ for $i,j=1,...,n$, $\sum_{j=1}^n k_{ij}=0$ for all $i$'s and $\sum_{i=1}^n k_i=1$. Using this fact together with \eqref{f3} and \eqref{sigmaij}, we may write the right hand side of \eqref{G} as
\begin{equation}
\sum_{j=1}^n k_{ij}\mathbf{J}_j=\sum_{j=1}^n k_{ij}  \int_{T} (\sigma^2(t)-\gamma(h_{0j})f(t))\mathbf{B}_{\lambda}(t)dt=-\sum_{j=1}^n k_{ij}\gamma(h_{0j})  \int_{T} f(t)\mathbf{B}_{\lambda}(t)dt, \hspace{10pt} i=1,...,n.
\label{kij}
\end{equation}
Thus, using \eqref{kij} and the expression for $\mathbf{G}$ we may now express \eqref{G} as
\begin{equation*}
\int_{T} f(t)\mathbf{B}_{\lambda}(t)\mathbf{B}_{\lambda}^{\intercal}(t)(\mathbf{b}_i-k_i\mathbf{c})dt=-\sum_{j=1}^n k_{ij}\gamma(h_{0j})  \int_{T} f(t)\mathbf{B}_{\lambda}(t)dt
\end{equation*}
or equivalently,
\begin{equation}
\label{2besolved}
\int_{T} f(t)\mathbf{B}_{\lambda}(t)   \mathbf{B}_{\lambda}^{\intercal}(t)\Big( \mathbf{b}_i-(k_i+\sum_{j=1}^n k_{ij}\gamma(h_{0j}))\mathbf{c}\Big)dt=\mathbf{0}, \hspace{10pt} i=1,...,n,
\end{equation}
where we in the last expression used the fact that $\mathbf{B}_{\lambda}^{\intercal}(t)\mathbf{c}$ must equal 1 for all values of $t\in T$ (the unbiasedness constraint). We now see that the above equation holds if $\mathbf{b}_i$ equals $(k_i+\sum_{j=1}^n k_{ij}\gamma(h_{0j}))\mathbf{c}$ and from the assumption of $\mathbf{G}^{-1}$'s existence we also have that this solution is the only one. Thus, the weights must equal
\begin{equation*}
\lambda_i(t)=\mathbf{b}^{\intercal}_i\mathbf{B}_{\lambda}(t)= (k_i+\sum_{j=1}^n k_{ij}\gamma(h_{0j}))\mathbf{c}^{\intercal} \mathbf{B}_{\lambda}(t)=k_i+\sum_{j=1}^n k_{ij}\gamma(h_{0j})=\lambda_i, 
\end{equation*} 
for all values of $t\in T$ and $i=1,...,n$, and consequently, the PWFK predictor coincides with the OKFD predictor.

\end{proof}

\section{Tables of the simulation study}
\label{AppendixTables}
Tables \ref{tab1}, \ref{tab3}, \ref{tab4} and \ref{tab6} contain information about the prediction performance in terms of mean squared prediction errors (MSPEs) for the simulated cases 1-18 for small and large sample sizes with and without a deterministic trend, respectively. The smallest overall MSPE for each case is highlighted in red. The numbers in parentheses represent the average computational time in seconds over the corresponding estimated models and replications. The column \#Times represents the number of times, out of the 100 realisations, that OKFD had lower (minimum) MSPE than the Sp.T. separable model. The last column shows p-values from two-sided paired t-tests comparing the overall MSPEs between the OKFD and the Sp.T. separable models.

\begin{table}[!htbp]
\scriptsize
\caption{Simulated data without deterministic time trend, small sample size.}

\centering
\begin{tabular}{ScScScc|cccc|cc}

   \hline\hline
\multicolumn{4}{c|}{Generated data} & \multicolumn{4}{c|}{overall MSPE} & \multicolumn{2}{c}{Comparison}\\

  \multicolumn{1}{c}{Scenario} & \multicolumn{1}{c}{Type}  & $\alpha$ & $\beta$ & OKFD & Sp.T. Separable & Sp.T. Product-sum & Sp.T. Metric & \#Times & p-value \\
  \hline \addlinespace

 1& \multicolumn{1}{c}{\multirowcell{12}{\rotatebox[origin=c]{90}{Separable}}} & & 0.1 & {\bf\textcolor{red}{0.068}} (0.2)  & 0.068 (7.1)  & 0.070 (9.8) & 0.106 (6.1) & 45 & 0.176\\
                                                    2& & 0.1 & 1 & 0.065 (0.2)  &  {\bf\textcolor{red}{0.065}} (8.4)  & 0.069 (10.8) & 0.080 (6.3)  & 39 & 0.561\\
                                                    3& & & 10 & {\bf\textcolor{red}{0.063}} (0.2)  & 0.064 (6.2)& 0.065 (14.6)  & 0.070 (6.3) & 37 & 0.409 \\ \addlinespace
                                                    
                                                    4& & & 0.1 & {\bf\textcolor{red}{0.135}} (0.2)   & 0.145 (7.0)  & 0.150 (12.7)  & 0.204 (6.4) & 68 &  $<$0.001  \\
                                                    5& & 0.5 & 1 & {\bf\textcolor{red}{0.135}} (0.2)   & 0.139 (8.6)  & 0.147 (14.6)  & 0.196 (6.1)  & 56 &  $<$0.001  \\
                                                    6& & & 10 & {\bf\textcolor{red}{0.134}} (0.2)  &  0.139 (7.3)  & 0.154 (17.3) & 0.154 (6.3)  & 53 & 0.204  \\ \addlinespace
                                                      
                                                    7& & & 0.1 &  {\bf\textcolor{red}{0.377}} (0.2)  & 0.400 (6.2)  & 0.395 (15.2) & 0.452 (6.2)  & 59 &  $<$0.001\\
                                                    8& & 2 & 1 &  {\bf\textcolor{red}{0.356}} (0.2)   & 0.386 (8.0)  & 0.399 (16.0)  & 0.476 (6.2)  & 67 & $<$0.001 \\
                                                    9& & & 10 & {\bf\textcolor{red}{0.365}} (0.2)  & 0.378 (7.8)  & 0.436 (17.6)  & 0.421 (6.2)  & 62 &  $<$0.001  \\ \addlinespace

\hline\addlinespace
 10& \multicolumn{1}{c}{\multirowcell{12}{\rotatebox[origin=c]{90}{Non Separable}}} &  & 0.1 &  {\bf\textcolor{red}{0.063}} (0.2)   & 0.063 (6.6) & 0.067 (9.0)& 0.105 (6.1) & 48 & 0.766 \\
                                                      11& & 0.1 & 1 & 0.063 (0.2)  & {\bf\textcolor{red}{0.063}} (8.5) & 0.066 (9.6) & 0.100 (6.3) & 35 & 0.755 \\
                                                      12& & & 10 & {\bf\textcolor{red}{0.063}} (0.2) & 0.065 (6.2) & 0.068 (10.4) & 0.090 (6.2) & 55 & 0.013   \\ \addlinespace
                                                      
                                                      13& & & 0.1 &  {\bf\textcolor{red}{0.132}} (0.2) & 0.144 (6.3)  & 0.153 (13.0) & 0.199 (6.3) &  64 &  $<$0.001 \\
                                                      14& & 0.5 & 1 & {\bf\textcolor{red}{0.132}} (0.2) & 0.140 (8.2)  & 0.150 (14.1) & 0.216 (5.9)  & 65 & $<$0.001  \\
                                                      15& & & 10 & {\bf\textcolor{red}{0.139}} (0.2) & 0.144 (9.0)  & 0.157 (16.6)  & 0.195 (6.4) & 62 &  $<$0.001 \\ \addlinespace
                                                      
                                                      16& & & 0.1 & {\bf\textcolor{red}{0.363}} (0.2)   & 0.398 (5.9)  & 0.379 (14.9) & 0.430 (6.1) & 67 &  $<$0.001  \\
                                                      17& & 2 & 1 & {\bf\textcolor{red}{0.365}} (0.2)   & 0.422 (7.1)  & 0.396 (16.0)  & 0.481 (6.2)  & 72 & $<$0.001 \\
                                                      18& & & 10 & {\bf\textcolor{red}{0.362}} (0.2) & 0.387 (8.2) & 0.385 (16.8) & 0.532 (6.5) & 70 &  $<$0.001  \\ \addlinespace

\hline

\end{tabular}
\label{tab1}
\normalsize
\end{table}

\begin{table}[!htbp]
\scriptsize
\caption{Simulated data without deterministic time trend, large sample size.}
\centering
\begin{tabular}{ScScScc|cc|cc}


   \hline\hline
\multicolumn{4}{c|}{Generated data} & \multicolumn{2}{c|}{overall MSPE} & \multicolumn{2}{c}{Comparison}\\
  \multicolumn{1}{c}{Scenario} & \multicolumn{1}{c}{Type}  & $\alpha$ & $\beta$ & OKFD & Sp.T. Separable & \#Times & p-value \\
  \hline \addlinespace

 1& \multicolumn{1}{c}{\multirowcell{12}{\rotatebox[origin=c]{90}{Separable}}} & & 0.1 & 0.051 (8.7) &{\bf\textcolor{red}{0.050}} (147.9) & 10 & $<$0.001\\
                                                    2& & 0.1 & 1 & 0.055 (8.6)   & {\bf\textcolor{red}{0.053}} (139.0) &  8 & $<$0.001\\
                                                    3& & & 10 & 0.054 (8.6)  & {\bf\textcolor{red}{0.051}} (125.8)  &  3 & $<$0.001\\ \addlinespace
                                                    
                                                    4& & & 0.1 & {\bf\textcolor{red}{0.078}} (9.9) & 0.080 (153.0) & 24 & 0.029  \\
                                                    5& & 0.5 & 1 & 0.079 (9.9)  & {\bf\textcolor{red}{0.078}} (154.5)  & 11 & 0.064 \\
                                                    6& & & 10 & 0.081 (10.0)  & {\bf\textcolor{red}{0.079}} (139.9) & 11 &  $<$0.001 \\ \addlinespace
                                                      
                                                    7& & & 0.1 & {\bf\textcolor{red}{0.164}} (9.6)  & 0.168 (146.4)  &  46 & 0.207\\
                                                    8& & 2 & 1 & {\bf\textcolor{red}{0.160}} (9.7)  & 0.161 (147.9)  &  38 & 0.147 \\
                                                    9& & & 10 &  0.167 (9.7)  & {\bf\textcolor{red}{0.167}} (141.3)  &  25 & 0.846 \\ \addlinespace

\hline\addlinespace
 10& \multicolumn{1}{c}{\multirowcell{12}{\rotatebox[origin=c]{90}{Non Separable}}} & & 0.1 & 0.053 (8.7)  & {\bf\textcolor{red}{0.052}} (150.1) & 4 & $<$0.001 \\
                                                      11& & 0.1 & 1 & 0.053 (8.7) & {\bf\textcolor{red}{0.051}} (146.5) & 2 & $<$0.001  \\
                                                      12& & & 10 & 0.054 (8.6)  & {\bf\textcolor{red}{0.050}} (134.6) &  3 &  $<$0.001 \\ \addlinespace
                                                      
                                                      13& & & 0.1 &  {\bf\textcolor{red}{0.076}} (9.3)   & 0.077 (156.4)  &  20 & 0.138  \\
                                                      14& & 0.5 & 1 & 0.078 (9.6)  & {\bf\textcolor{red}{0.077}} (163.3)   & 14 & 0.431 \\
                                                      15& & & 10 & 0.078 (10.3) &  {\bf\textcolor{red}{0.075}} (177.5)  &  14 & $<$0.001 \\ \addlinespace
                                                      
                                                      16& & & 0.1 &  {\bf\textcolor{red}{0.160}} (9.7)  & 0.175 (145.6) & 41 & 0.038 \\
                                                      17& & 2 & 1 & {\bf\textcolor{red}{0.161}} (9.7)   & 0.162 (147.7)  &  22 & 0.516 \\
                                                      18& & & 10 & 0.165 (9.7)  & {\bf\textcolor{red}{0.164}} (148.3)  & 17 & 0.723  \\ \addlinespace

\hline
\end{tabular}
\label{tab3}
\normalsize
\end{table}

\newpage
\begin{table}[!htbp]
\scriptsize
\caption{Simulated data with deterministic time trend, small sample size.}

\centering
\begin{tabular}{ScScScc|cccc|cc}

   \hline\hline
\multicolumn{4}{c|}{Generated data} & \multicolumn{4}{c|}{overall MSPE} & \multicolumn{2}{c}{Comparison}\\

  \multicolumn{1}{c}{Scenario} & \multicolumn{1}{c}{Type}  & $\alpha$ & $\beta$ & OKFD & Sp.T. Separable & Sp.T. Product-sum & Sp.T. Metric & \#Times & p-value \\
  \hline \addlinespace

 1& \multicolumn{1}{c}{\multirowcell{12}{\rotatebox[origin=c]{90}{Separable}}} & & 0.1 & 0.066 (0.2) & {\bf\textcolor{red}{0.066}} (9.3) & 0.068 (12.6)& 0.072 (6.2) & 31 & 0.107  \\
                                                    2& & 0.1 & 1 & 0.063 (0.2) &{\bf\textcolor{red}{0.063}} (8.8)& 0.066 (11.9) & 0.070 (6.1) & 29  & 0.004\\
                                                    3& & & 10 & 0.066 (0.2)& {\bf\textcolor{red}{0.066}} (6.7) & 0.068 (15.5) & 0.066 (6.3) & 32 & 0.441  \\ \addlinespace                                                    
                                                      
                                                    4& & & 0.1 &  {\bf\textcolor{red}{0.128}} (0.2) & 0.134 (9.7) & 0.138 (14.6) & 0.149 (6.4) & 69 & $<$0.001   \\
                                                    5& & 0.5 & 1 & {\bf\textcolor{red}{0.135}} (0.2)  & 0.140 (9.9)  & 0.148 (14.5) & 0.146 (5.9) & 63 & 0.029  \\
                                                    6& & & 10 & {\bf\textcolor{red}{0.137}} (0.2) & 0.138 (7.8)  & 0.159 (17.3) & 0.146 (6.4) & 49 & 0.022  \\ \addlinespace
                                                      
                                                    7& & & 0.1 &  {\bf\textcolor{red}{0.380}} (0.2)   & 0.398 (8.2) & 0.413 (15.5) & 0.457 (6.2) & 73 & $<$0.001  \\
                                                    8& & 2 & 1 &   {\bf\textcolor{red}{0.377}} (0.2)  & 0.400 (9.4) & 0.422 (15.6) & 0.430 (5.9) & 65 & $<$0.001 \\
                                                    9& & & 10 & {\bf\textcolor{red}{0.387}} (0.2)  & 0.405 (8.1)  & 0.442 (17.6) & 0.415 (6.2) & 73 &  $<$0.001 \\ \addlinespace

\hline\addlinespace
 10& \multicolumn{1}{c}{\multirowcell{12}{\rotatebox[origin=c]{90}{Non Separable}}} &  & 0.1 & {\bf\textcolor{red}{0.063}} (0.2)  & 0.063 (9.6)  & 0.064 (12.7) & 0.066 (6.3) & 34 & 0.598 \\

                                                      11& & 0.1 & 1 & 0.065 (0.2) & {\bf\textcolor{red}{0.064}} (9.0) & 0.066 (11.1) & 0.066 (6.2) & 33 &  0.009 \\
                                                      12& & & 10 & 0.066 (0.2) &   {\bf\textcolor{red}{0.065}} (6.7)  & 0.071 (13.2)& 0.069 (6.3) & 30 & 0.049 \\ \addlinespace
                                                      
                                                      13& & & 0.1 &  {\bf\textcolor{red}{0.135}} (0.2)   & 0.139 (9.7) & 0.148 (14.5) & 0.153 (6.5) & 63 & $<$0.001 \\
                                                      14& & 0.5 & 1 & {\bf\textcolor{red}{0.133}} (0.2) & 0.136 (10.0)  & 0.146 (14.2) & 0.153 (5.9) & 67 & $<$0.001   \\
                                                      15& & & 10 &   {\bf\textcolor{red}{0.134}} (0.2) & 0.136 (8.9) & 0.148 (16.6) & 0.142 (6.2) & 58 &  $<$0.001  \\ \addlinespace
                                                      
                                                      16& & & 0.1 & {\bf\textcolor{red}{0.362}} (0.2)  & 0.375 (8.3)   & 0.384 (15.4) & 0.447 (6.4) & 69 &  $<$0.001\\ 
                                                      17& & 2 & 1 &   {\bf\textcolor{red}{0.360}} (0.2)  & 0.377 (9.4) & 0.384 (15.9) & 0.431 (6.1) & 74 & $<$0.001 \\
                                                      18& & & 10 & {\bf\textcolor{red}{0.377}} (0.2)  &  0.402 (8.9)  & 0.412 (16.8)  & 0.423 (6.1) & 81 &  $<$0.001\\ \addlinespace

\hline

\end{tabular}
\label{tab4}
\normalsize
\end{table}

\begin{table}[!htbp]
\scriptsize
\caption{Simulated data with deterministic time trend, large sample size.}
\centering
\begin{tabular}{ScScScc|cc|cc}


   \hline\hline
\multicolumn{4}{c|}{Generated data} & \multicolumn{2}{c|}{overall MSPE} & \multicolumn{2}{c}{Comparison}\\
  \multicolumn{1}{c}{Scenario} & \multicolumn{1}{c}{Type}  & $\alpha$ & $\beta$ & OKFD & Sp.T. Separable & \#Times & p-value \\
  \hline \addlinespace

 1& \multicolumn{1}{c}{\multirowcell{12}{\rotatebox[origin=c]{90}{Separable}}} & & 0.1 & 0.054 (8.7) & {\bf\textcolor{red}{0.053}} (151.7) & 10 & $<$0.001\\
                                                    2& & 0.1 & 1 & 0.053 (8.7)   & {\bf\textcolor{red}{0.051}} (143.9) &  8 & $<$0.001\\
                                                    3& & & 10 & 0.054 (8.7)  & {\bf\textcolor{red}{0.051}} (131.9)  &  3 & $<$0.001\\ \addlinespace
                                                    
                                                    4& & & 0.1 & {\bf\textcolor{red}{0.074}} (9.9) & 0.076 (158.6) & 24 & 0.016  \\
                                                    5& & 0.5 & 1 & 0.077 (9.9)  & {\bf\textcolor{red}{0.077}} (159.5)  & 11 & 0.465 \\
                                                    6& & & 10 & 0.080 (9.9)  & {\bf\textcolor{red}{0.077}} (146.3) & 11 &  $<$0.001 \\ \addlinespace
                                                      
                                                    7& & & 0.1 & {\bf\textcolor{red}{0.156}} (9.6)  & 0.164 (150.8)  &  46 & 0.151\\
                                                    8& & 2 & 1 & {\bf\textcolor{red}{0.163}} (9.7)  & 0.169 (150.5)  &  38 & 0.104 \\
                                                    9& & & 10 &  {\bf\textcolor{red}{0.165}} (10.0)  & 0.166 (153.1)  &  25 & 0.479 \\ \addlinespace

\hline\addlinespace
 10& \multicolumn{1}{c}{\multirowcell{12}{\rotatebox[origin=c]{90}{Non Separable}}} & & 0.1 & 0.053 (8.7)  & {\bf\textcolor{red}{0.051}} (151.6) & 4 & $<$0.001 \\
                                                      11& & 0.1 & 1 & 0.053 (8.7) & {\bf\textcolor{red}{0.050}} (150.0) & 2 & $<$0.001  \\
                                                      12& & & 10 & 0.055 (8.6)  & {\bf\textcolor{red}{0.051}} (134.1) &  3 &  $<$0.001 \\ \addlinespace
                                                      
                                                      13& & & 0.1 &  {\bf\textcolor{red}{0.077}} (10.3)   & 0.078 (190.2)  &  20 & 0.051  \\
                                                      14& & 0.5 & 1 & 0.077 (9.9)  & {\bf\textcolor{red}{0.076}} (159.8)   & 14 & 0.288 \\
                                                      15& & & 10 & 0.079 (10.0) &  {\bf\textcolor{red}{0.077}} (155.8)  &  14 & $<$0.001 \\ \addlinespace
                                                      
                                                      16& & & 0.1 &  {\bf\textcolor{red}{0.163}} (9.7)  & 0.166 (150.7) & 41 & 0.029 \\
                                                      17& & 2 & 1 & {\bf\textcolor{red}{0.163}} (9.7)   & 0.167 (152.6)  &  22 & 0.186 \\
                                                      18& & & 10 & 0.163 (9.7)  & {\bf\textcolor{red}{0.162}} (152.5)  & 17 & 0.213  \\ \addlinespace

\hline
\end{tabular}
\label{tab6}
\normalsize
\end{table}

\newpage
Tables \ref{Nbasisstat} and \ref{Nbasisnonstat} contain information about the number of basis functions used in the OKFD models concerning the stationary (corresponding to isotropic Sp.T. processes with separable and non-separable covariance function) and non-stationary scenarios for the considered sample sizes.

\begin{table}[!htbp]
\footnotesize
\caption{Stationary}
\centering
\begin{tabular}{cccc}

   \hline\hline

  \multicolumn{1}{c}{} & \multicolumn{3}{c}{Sample sizes} \\
  &Small&Medium&Large\\
  \hline \addlinespace                                         
  Fourier&5,7,9,11&5,15,25,35,45,47,49&5,15,25,35,45,47,49\\
  B-splines&5,6,7,8,9,10,11,12&5,15,25,35,45,47,49&5,15,25,35,45,47,49\\ \addlinespace
  Total&12&14&14\\
  \addlinespace\hline
  
\end{tabular}
\label{Nbasisstat}
\normalsize
\end{table}

\begin{table}[!htbp]
\footnotesize
\caption{Non-stationary}

\centering
\begin{tabular}{cccc}

   \hline\hline

  \multicolumn{1}{c}{} & \multicolumn{3}{c}{Sample sizes} \\
  &Small&Medium&Large\\
  \hline \addlinespace                                         
  Fourier&5,7,9,11&5,7,9,11,13,15,17,19,21,23,25,35,45,47,49&5,7,9,11,13,15,17,19,21,23,25,35,45,47,49\\
  B-splines&5,6,7,8,9,10,11,12&5,7,9,11,13,15,17,19,21,23,25,35,45,47,49&5,7,9,11,13,15,17,19,21,23,25,35,45,47,49\\ \addlinespace
  Total&12&30&30\\
  \addlinespace\hline
  
\end{tabular}
\label{Nbasisnonstat}
\normalsize
\end{table}



\newpage

\bibliographystyle{apalike} 
\bibliography{newRef}

\end{document}